\documentclass[11pt]{article}

\usepackage{xspace}             % \xspace
\usepackage{hyperref}

\usepackage{amsmath}
\usepackage{amsthm}
\usepackage{amsfonts}
\usepackage{amssymb}

\usepackage{authblk} % \affil ?
\usepackage{caption}
\usepackage{subcaption}
\usepackage{xcolor,cancel}

\usepackage{graphicx}           % \includegraphics 
\usepackage{import}             % \subimport
\usepackage{epstopdf}           % for included gnuplots

\usepackage{cleveref}           % \Cref \cref

\newtheorem{theorem}{Theorem}
\newtheorem{lemma}{Lemma}
\newtheorem{corollary}{Corollary}
\newtheorem{definition}{Definition}

\newtheorem{sublemma}{Lemma}[section]
\newtheorem{subdefinition}{Definition}[section]

\newcommand{\setNealCounter}[1]{\setcounter{#1}{\value{section}}\addtocounter{#1}{-1}}
\newcommand{\setNealCounters}
{\setNealCounter{theorem}\setNealCounter{lemma}\setNealCounter{corollary}\setNealCounter{definition}}

\newcommand{\E}{\operatorname{E}}

\newcommand{\opt}{\ensuremath{\operatorname{\mbox{\sc opt}}}\xspace}
\newcommand{\Opt}{\ensuremath{\operatorname{\mbox{\sc Opt}}}\xspace}

\newcommand{\lb}{\ensuremath{\operatorname{\mbox{\sc lb}}}\xspace}

\newcommand{\brb}{\ensuremath{\operatorname{\mbox{\sc brb}}}\xspace}
\newcommand{\Brb}{\ensuremath{\operatorname{\mbox{\sc Brb}}}\xspace}

\newcommand{\brbx}[1]{\ensuremath{\brb_{#1}}\xspace}
\newcommand{\Brbx}[1]{\ensuremath{\Brb_{#1}}\xspace}

\newcommand{\brbk}{\brbx{K}}
\newcommand{\Brbk}{\Brbx{K}}

\newcommand{\bmc}{\ensuremath{\operatorname{\mbox{\sc bmc}}}\xspace}
\newcommand{\Bmc}{\ensuremath{\operatorname{\mbox{\sc Bmc}}}\xspace}
\newcommand{\bmcf}{\ensuremath{\bmc_f}\xspace}

\newcommand{\bmcfd}{\ensuremath{\bmc_{f_d}}\xspace}

\newcommand{\linearBmc}{{\sc linear} \bmc}
\newcommand{\latencyBmc}{\ensuremath{\bmc_{\le K}}\xspace}
\newcommand{\LinearBmc}{{\sc Linear} \bmc}
\newcommand{\LatencyBmc}{\ensuremath{\Bmc_{\le K}}\xspace}

\newcommand{\newfn}[1]{\operatorname{\sf #1}}

\newcommand{\cost}{\newfn{cost}}

\newcommand{\latency}{\newfn{latency}}

\newcommand{\leftDepth}{\newfn{left\_depth}}
\newcommand{\rightDepth}{\newfn{right\_depth}}

\newcommand{\rbar}{{\overline r}}
\newcommand{\lbar}{{\overline \ell}}

\newcommand{\calI}{\ensuremath{{\cal I}}\xspace}

\newcommand{\eps}{\epsilon}

\newcommand{\Rp}{\mathbb{R}_{+}}

  \title{Online Bigtable merge compaction}

\newcommand{\AFFIL}[2]{\affil{\protect\makebox[5in]{#1\hfill \small\tt #2}}}

\author{Claire Mathieu}
\AFFIL{CNRS, Paris}{cmathieu\,@\,di.ens.fr}
\author{Carl Staelin}
\AFFIL{Google Israel Engineering Center, Haifa}{staelin\,@\,google.com}
\author{Neal E. Young\,\thanks{Supported
    by Google research award
    \emph{A Study of Online Bigtable-Compaction Algorithms}, and NSF grant 1117954.}}
\AFFIL{University of California, Riverside}{neal.young\,@\,ucr.edu}
\author{Arman Yousefi\,\protect\footnotemark[1]}
\AFFIL{University of California, Los Angeles}{armany@cs.ucla.edu}

\date{}

\begin{document}

\maketitle
\begin{abstract}
  NoSQL databases are widely used
  for massive data storage and real-time web applications.
  Yet important aspects of these data structures are not well understood.
  For example, NoSQL databases write most of their data to a collection
  of files on disk, meanwhile periodically \emph{compacting} subsets of these files.
  A \emph{compaction policy} must choose which files to compact, and when to compact them,
  without knowing the future workload.
  Although these choices can affect computational efficiency by orders of magnitude,
  existing literature lacks tools for designing and analyzing online compaction policies ---
  policies are now chosen largely by trial and error.
  % Here we initiate the formal study of an online optimization problem
  % --- \emph{Bigtable compaction} ---
  % that models choosing Bigtable compactions. 

  Here we introduce tools for the design and analysis of compaction policies for Google Bigtable,
  propose new policies, give average-case and worst-case competitive analyses,
  and present preliminary empirical benchmarks.
  
 \end{abstract}

\thispagestyle{empty}
\titlepage

%\paragraph{Problem statement.}

%\setcounter{section}{-1}
\section*{Introduction 
  --- NoSQL databases and BigTable compaction}%\label{sec:introduction}
NoSQL databases provide distributed, reliable, high-volume, real-time data storage.
Companies making heavy use of NoSQL systems include
Adobe, Ebay, Facebook, GitHub, Meetup, Netflix, and Twitter.
At Google, \emph{BigTable} servers support applications
such as Gmail, Maps, Search, Crawl, Google+, Analytics, and Base.
Published data (most recently from 2006) show over 24,500 BigTable servers,
supporting over 1.2 million requests per second and 16 GB/s of outgoing RPC traffic,
and holding over a petabyte of data for Google Crawl and Analytics alone~\cite[\S 8]{chang2008bigtable}.

% \marginpar{\small \bf \raggedright stack based: 
%   Accumulo, % http://accumulo.apache.org/1.7/accumulo_user_manual.html#_compaction
%   AsterixDB, % from vagelis, not well documented online
%   HBase, % http://blog.cloudera.com/blog/2013/12/what-are-hbase-compactions/
%   Hypertable, % http://hypertable.org says "like bigtable"
%   Spanner, % from carl, not well documented
%   others?}
% \marginpar{\small \bf \raggedright not stack based: 
%   Cassandra % http://www.datastax.com/dev/blog/leveled-compaction-in-apache-cassandra
%   LevelDB, % e.g. http://leveldb.googlecode.com/svn/trunk/doc/impl.html
%   others?}
For a general introduction to NoSQL,
see~\cite{cattell2011scalable,redmond2012seven,strauch2011nosql}. 
Roughly, NoSQL databases support reads and writes of key/value pairs.
Almost all modern NoSQL systems employ a ``Log-Structured-Merge'' (LSM) architecture:
a cache holds recent writes, which are periodically aggregated and pushed to immutable disk files.
This is in contrast to traditional DBMSs, 
which update data files in place, leading to slower insertions and updates.
LSM systems organize their files in \emph{levels}
by partitioning time into intervals and storing all writes from a particular interval in one level.
The most recent level (ending at the current time) is held in the cache.
Each remaining level is held on disk, either in a single file
or, by a partition of the key space, in multiple files.
Periodically, the cache is dumped to disk, creating a new level.
(The cache may be dumped for various reasons, not just when it is full.)
The time per read grows with the number of levels
--- a typical read searches the levels, most recent first,
checking one file in each level until the desired key is found.
To keep the number of levels bounded, contiguous levels are periodically merged.
This merge process is referred to as \emph{compaction}.
Compaction and read operations together
account for a significant fraction of the computing resources used by the system, 
and can be the main bottleneck~\cite[\S 7]{chang2008bigtable}.

Here we focus on improving the efficiency of compaction and reads.
We focus on Google's BigTable database,
but the proposed principles may also be applied to other LSM storage systems,
most immediately to those that, like Bigtable, use just one file per level
(e.g.~Accumulo~\cite{kepner_achieving_2014,patil_ycsb++:_2011},
AsterixDB~\cite{alsubaiee2014asterixdb}, 
HBase~\cite{patil_ycsb++:_2011,george_hbase:_2011,khetrapal_hbase_2006},
Hypertable~\cite{khetrapal_hbase_2006,judd_scale_2008}, and
Spanner~\cite{corbett2013spanner}).
We develop techniques for the design and analysis of compaction policies,
analyze new policies using worst-case and average-case competitive analyses,
give absolute estimates of optimal costs, and present preliminary benchmarks.

\smallskip

This is the first formal study of online compaction policies that we know of.\footnote
{Ghosh et al.~study the related but quite different problem of 
performing a single offline compaction via a sequence of merges,
given a constraint on the number of files that can be merged at once.
That problem is NP-hard~\cite{ghosh_fast_2015}.
As far as we know, NoSQL is not yet studied in the large literature on
\emph{external-memory} algorithms~\cite{arge_external-memory_2010,vitter_external_2001}.}

\paragraph{Formal definition of Bigtable merge compaction (BMC).}
Formally, for any non-decreasing \emph{read-cost function} $f:\Rp\rightarrow\Rp$, 
define \bmcf as follows.
The input is a sequence $\calI = \langle (\ell_t,r_t)\rangle_t \in (\Rp\!\times\Rp)^n$.
The algorithm maintains a stack of lengths, initially empty.
At time $t$, the pair $(r_t, \ell_t)$ is revealed,
where $r_t$ is the \emph{read rate} and $\ell_t$ is the length at time $t$
(representing the length of the new disk file created from a cache dump).
The length $\ell_t$ is inserted at the top of the stack.
The algorithm $A$ then chooses a \emph{compaction}:
it selects some contiguous sequence of lengths at the top of the stack, 
then adds them to get a single new length $L_t$, which replaces them in the stack.
At time $t$, the \emph{merge cost} is $L_t$; the \emph{read cost} is $r_t\, f(k_t)$,
where $k_t$ is the stack size after the compaction at time $t$.
The output, called a \emph{schedule}, 
is the sequence $\sigma$ of $n$ compactions.
The cost of $\sigma$ on $\calI$, denoted $\sigma(\calI)$ or $A(\calI)$, is
$\sum_{t=1}^n L_t + r_t\,f(k_t)$.
\Cref{fig:example} shows an example schedule. 

\begin{figure}
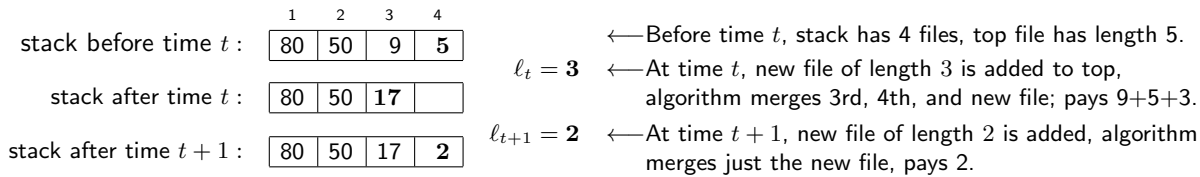
\sf\footnotesize
  \newcommand{\tmpA}[1]{\longleftarrow & \text{\footnotesize #1}}
  \newcommand{\tmpB}[1]{& \text{\footnotesize #1}}
  \newcommand{\SLOT}[2]{%
    \lefteqn{\raisebox{12pt}[0pt][0pt]{\tiny~#1}}\raisebox{-0.5pt}{\makebox[0.9em][r]{#2}}
  }
  \newcommand{\SLOTSS}[4]{%
    \begin{array}{|c|c|c|c|} \hline
     \SLOT{1}{#1} & \SLOT{2}{#2} &  \SLOT{3}{#3} & \SLOT{4}{#4} \\ \hline 
    \end{array}
  }
 \newcommand{\SLOTS}[4]{%
   \begin{array}{|c|c|c|c|} \hline
     \SLOT{}{#1} & \SLOT{}{#2} &  \SLOT{}{#3} & \SLOT{}{#4} \\ \hline 
   \end{array}
 }
 \vspace*{-0.1in}
 \hspace*{-1.5em}
 \begin{equation*}
   \begin{array}{rrlll@{}l}~
     &\small \text{stack before time } t: & \SLOTSS{80}{50}{9}{\bf 5} 
     \\[-15pt]
     &&&& \tmpA{Before time $t$, stack has 4 files, top file has length 5.}

     \\[2pt] &&& {~~~\ell_t = \mathbf 3}
     &\tmpA{At time $t$, new file of length $3$ is added to top,}
     \\ &&&&\tmpB{algorithm merges 3rd, 4th, and new file; pays 9+5+3.}

     \\[-12pt]& \text{stack after time } t:&\SLOTS{80}{50}{\bf 17}{}

     \\[3pt] &&&{\ell_{t+1} = \mathbf 2}
     &\tmpA{At time $t+1$, new file of length $2$ is added, algorithm}
     \\ &&&&\tmpB{merges just the new file, pays 2.}

     \\[-16pt] &\text{stack after time } t+1:&\SLOTS{80}{50}{17}{\bf 2}
   \end{array}
 \end{equation*}
 \vspace*{-0.1in}
 \caption{\sf Steps $t$ and $t+1$ of a \bmcf schedule.}\label{fig:example}
 \vspace*{-0.1in}
\end{figure}
%\marginpar{``FILES'' should be ``LENGTHS''}

Current practice at Google is to constrain the number of levels to a parameter $K$,
otherwise ignoring read costs.  We use \latencyBmc to denote this special case of \bmcf,
which is obtained by taking $f(k) = 0$ if $k\le K$ and $f(k)=\infty$ otherwise.
The parameter $K$ is tuned manually on a per-table basis, based on historical workload. 
This is reliable, but slow, costly, and inflexible.
% in the face of tables in which different regions may have different access patterns~\cite{carl}.
To explore compaction policies that instead adjust stack size \emph{automatically},
we also consider \linearBmc, which is \bmcf with $f(k)=k$.

\smallskip

For more intuition about the combinatorial structure of \bmcf,
note that the restriction of \latencyBmc to \emph{uniform} instances
(those with $(\ell_t,r_t)=(\lbar,\rbar)$ for all $t$)
is essentially the \emph{egg-dropping puzzle} 
with $n$ floors and $K$ eggs~\cite[Thm.~2]{sniedovich2003or}
(\cite{bentley_general_1982} gives other applications).
The restriction of \linearBmc to uniform instances
is equivalent to \emph{lopsided alphabetic binary coding}~\cite{choy_construction_1983,golin_more_2008,kapoor_optimum_1989}.
We encourage the reader to try solving a uniform instance of \latencyBmc
with $n$ unit lengths and, say, $K=1$ and then $K=2$.
Uniform instances are already combinatorially non-trivial;
the general cases with non-uniform inputs are significantly more complicated.

\smallskip 

Throughout, $X\sim Y$ means $X=(1\pm o(1))Y$,
where $o(1)$ denotes a quantity that tends to zero 
as $n = |\calI|$ tends to infinity.
\emph{With high probability} means with probability $1-o(1)$,
and $[i,j]$ denotes $\{i,i+1,\ldots,j\}$.
$\calI[i,j]$ denotes $(\ell_i,r_i),(\ell_{i+1},r_{i+1}),\ldots,(\ell_j,r_j)$.
A compaction algorithm $A$ is \emph{online} if its choice at time $t$ 
depends only on $\calI[1,t]$.
$A$ is \emph{$c$-competitive} if $A(\calI)\le c\, \opt(\calI)$
for every instance $\calI$.
Given a random instance $\calI$, $A$
is \emph{$c$-competitive in expectation} if $\E_{\calI}[A(\calI)] \le c\,\E_\calI[\opt(\calI)]$,
and \emph{asymptotically 1-competitive in expectation}
if  $\E_{\calI}[A(\calI)] \sim \E_\calI[\opt(\calI)]$.

\paragraph{Summary of main theorems}

\begin{description}\itemsep0em\vspace*{-5pt}

\item[\Cref{thm:latency worst case} (worst-case analysis of {\footnotesize BMC$_{\le \mathbf K}$)}.]
  \emph{There is an online algorithm (called \brb) for \latencyBmc that is $K$-competitive.
    No deterministic online algorithm is less than $K$-competitive.}

% \item[\S\ref{sec:bijection}] \Cref{thm:bijection} shows that
%   \emph{for any length-$n$ instance $\calI$ of \bmcf, 
%     the schedules $\sigma$ for $\calI$ are isomorphic 
%     to $n$-node binary search trees $T$
%     (and vice versa) under a natural cost function.}

\item[\Cref{thm:bijection} (bijection with binary search trees).]
  \emph{For any instance $\calI$ of \bmcf, 
    the schedules $\sigma$ for $\calI$ are isomorphic 
    to the $n$-node binary search trees $T$, under a natural cost function\ldots}

\item[\Cref{thm:linear worst case} (worst-case analysis of {\footnotesize LINEAR BMC}).]
  \emph{There is an online algorithm for \linearBmc that is
    $O(1)$-competitive on ``read-heavy'' instances $\calI$ ---
    those s.t.~$\ell_t = O(r_t)$ for all $t$.}

\item[\Cref{thm:average case} (average-case analyses).]
  \emph{\latencyBmc and \linearBmc have online algorithms $A$ and $B$, respectively,
    that are asymptotically 1-competitive in expectation
    on random inputs $\calI$ with bounded, i.i.d.~requests.
    % --- that is, any input $\calI$
    % where the $n$ pairs $(\ell_t,r_t)$ are independently and identically distributed (i.i.d.)
    % from some bounded distribution.
    % 
    %\\[2pt]
    On such an \calI, 
    letting $(\lbar,\rbar) =(\E_\calI[\ell_t],\E_\calI[\ell_t])$ (for all $t$),
    for \latencyBmc, 
    \[\E_{\calI}[A(\calI)]\,\sim\, \E_{\calI}[\opt(\calI)]\, \sim\, \lbar\, K n^{1+1/K}/c_K\]
    where $c_K = (K+1)/(K!)^{1/K}$ (so $c_K\rightarrow e$ for large $K$). 
    For \linearBmc, 
    \[\E_{\calI}[B(\calI)] \,\sim\, \E_{\calI}[\opt(\calI)] \,\sim\, \beta_\calI\, n \log_2 n,~~~~\,\]
    for
    $\beta_\calI=\beta$ 
    such that
    $1/2^{\beta/\lbar} + 1/2^{\beta/\,\rbar} = 1$,
    so
    \(\beta
    = \Theta(\lbar + \rbar)
    /
    \ln\, (1\!+\!\max(\lbar/\rbar,\, \rbar/\lbar))
    \). 
  }
\end{description}

\paragraph{Benchmarks.}
In many applications at Google, the lengths of inserted files (the $\ell_t$'s) 
follow $\log$-normal distributions.
\Cref{sec:benchmarks} presents empirical benchmarks on such distributions.
The algorithm from \Cref{thm:latency worst case},
\brb\ --- balanced rent-or-buy, performs nearly optimally,
better (sometimes substantially)
than the current default BigTable compaction algorithm (for \latencyBmc).

\paragraph{Techniques.}
\Brb, our $K$-competitive algorithm for \latencyBmc,
is a recursive rent-or-buy scheme that
roughly balances the cost incurred in each of the $K$ stack positions.
\Brb happens to be asymptotically \emph{optimal} on uniform instances.
The proof of $K$-competitiveness is by induction on $K$.
The proof that no algorithm is better than $K$-competitive
uses a non-trivial recursive generalization of the standard rent-or-buy adversary argument.

Offline \bmcf has straightforward dynamic-programming algorithms ---
$O(n^4)$ time for \bmcf, $O(K n^3)$ for $\latencyBmc$, $O(n^3)$ for \linearBmc (\Cref{cor:alg}).
\Cref{thm:bijection} (the bijection with binary trees) is the critical observation 
that unlocks \linearBmc for further analysis.
The theorem yields a tree-based lower bound on \opt (\Cref{thm:lower bound})
analogous to entropy-based lower bounds for alphabetic codes~\cite{golin_more_2008}. 
The lower bound in turn is used to give a linear-time 2-approximation algorithm for \linearBmc 
(\Cref{cor:2-approx}),
and to bound \opt in the proof of \Cref{thm:linear worst case}.

\Cref{thm:bijection} is also used in the proof of \Cref{thm:average case}:
firstly, to bound optimal solutions for \emph{uniform} instances $\overline\calI$
(which correspond exactly to optimal binary search trees and alphabetic codes,
whose costs are well understood);
secondly, to show that, with high probability,
random instances $\calI$ and uniform instances have the same asymptotic cost.

\paragraph{Remarks.}
One aspect of compaction not modeled by \bmcf as defined here
is that key/value pairs may leave the database, due to expiration, deletion, or redundancy.
When a compaction merges several files into one file $F$,
the length of $F$ may be \emph{less than} the length of the merged files.
We note without proof that the $K$-competitive algorithm \brb for \latencyBmc 
(and its proof) extend naturally to show $K$-competitiveness in this more general setting.

It is natural to extend \bmcf to allow so-called  \emph{interior merges}, 
which merge contiguous levels \emph{within} the stack.  
\Opt never uses interior merges,
nor does \brb
(which remains optimally $K$-competitive for \latencyBmc even if interior merges are allowed).
But we conjecture that any $O(1)$-competitive online algorithm for general \linearBmc
will require interior merges.

We're conducting further benchmarks using AsterixDB,
after which we'll benchmark on Google BigTable servers.
Many theoretical problems remain open.
Is \brb asymptotically 1-competitive in expectation on bounded i.i.d.~inputs?
Is there an $o(K)$-competitive \emph{randomized} online algorithm for $\latencyBmc$?
Is there an $O(1)$-competitive online algorithm for general \linearBmc?

% \medskip

% \centerline{\bf LIST OF TODOS}

% \begin{enumerate}\itemsep0in
% \item add analytic estimates of Google Default on uniform inputs.
%   $\Theta(n^2/2^k)$ (differs from previous!)

% \item Finish benchmarks.  Compare to opt in body?  Read costs?
% \item Add some nice pictures to the full paper?
% \end{enumerate}

%%% Local Variables: 
%%% mode: latex
%%% TeX-master: "main"
%%% End: 

\section{Worst-case competitive analysis of BMC$_{\le K}$}
\label{sec:latency worst case}
\setNealCounters

\paragraph{Definition of algorithm \brbk for  \latencyBmc on input \calI.}
For $K=1$, there is only one possible schedule: at each time $t$, all files are merged into one.
For $K>1$, \brbk partitions the times $[1,n]$ into intervals called \emph{phases}.
The first phase $[1,1]$ starts and ends at time $1$.
Each subsequent phase $[s,s']$ ends with \brbk merging all files into one file at time $s'$.
To handle the requests in $[s,s'-1]$ (before the end of the phase),
\brbk runs $\brb_{K-1}$ recursively,
ignoring the single file at the bottom of the stack from the previous phase.
The phase is as long as possible, subject to the constraint that
the cost that $\brb_{K-1}$ incurs during the phase, $\brb_{K-1}(\calI[s,s'])$,
is less than $K-1$ times the cost of the single merge that \brbk does to end the phase, $\ell[1,s']$.
(See (a) in the proof below.)

\begin{theorem}[worst-case analysis for \latencyBmc]\label{thm:latency worst case} 
  (i) \Brbk is $K$-competitive for \latencyBmc. 
  \\[-16pt] 
  \begin{enumerate}\itemsep0in
  \item[(ii)] No deterministic online algorithm for \latencyBmc is less than $K$-competitive.
  \end{enumerate}
\end{theorem}
The proof consists of the two lemmas below.

\begin{sublemma}[Part (i)]\label{lemma:latency worst case upper bound}
  There exists a $K$-competitive online algorithm for \latencyBmc. 
\end{sublemma}
\begin{proof}
  Fix an input $\calI$.  
  Let $\calI[i,j]$ denote the subsequence $(\ell_i,r_i),\ldots,(\ell_j,r_j)$ of $\calI$.
  Let $\ell[i,j] = \sum_{h=i}^j \ell_h$. 
  For $K=1$, all algorithms are the same, hence 1-competitive.
  To complete the proof, for $K>1$, we show that, for each phase $[s,s']$, 
  during the phase, the cost incurred by $\brb_K$ 
  is at most $K$ times the cost incurred by \opt.
  First consider any phase that ends with $\brb_K$ merging all files into one
  (as happens in every phase except maybe the last).
  During the phase:
  \\[8pt]
  \begin{tabular}{r@{ }l@{~~}l}
    (a) &  $\Brb_K$ chooses $s'$ so 
    & \(
      \brb_{K-1}(\calI[s,s'-1]) \,<\, (K-1)\,\ell[1,s'] \,\le\, \brb_{K-1}(\calI[s,s']).
      \) 
    \\[5pt]
    (b) &  $\brb_K$ incurs cost
    & \( \brb_{K-1}(\calI[s,s'-1]) \,+\, \ell[1,s'] \).
    \\[5pt]
    (c) & \Opt incurs cost at least 
    & \(  \min\big\{\textstyle\frac{1}{K-1}\brb_{K-1}(\calI[s,s']), ~ \ell[1,s']\big\} \).
      {}~~~~(This is proven below.)
  \end{tabular}
  \\[8pt]
  Bounds (a-c) above imply, by algebra, that \brbk's cost during the phase 
  is at most $K$ times \opt's cost during the phase.
  The proof of (c) has two cases:
  \begin{description}\itemsep0in
  \item[\rm\emph{\Opt merges all files into one at some time $t\in [s,s']$.}]
    For that merge \opt pays $\ell[1,t]$.
    At each time $t'\in[t+1,s']$ \opt pays at least $\ell_{t'}$.
    \Opt's total cost during the phase is at least $\ell[1,s']$.

  \item[\rm\emph{\Opt never merges all files into one during $[s,s']$.}]
    Whatever file \opt had at the bottom of the stack at time $s$ remains
    untouched throughout the phase.
    Hence, \opt handles $\calI[s,s']$ using only $K-1$ stack slots.
    By induction, $\brb_{K-1}$ is $(K-1)$-competitive on $\calI[s,s']$, 
    so \opt's cost to do so is at least $\brb_{K-1}(\calI[s,s'])/(K-1)$.
    % (At time $s$, \opt may have files in the $K-1$ slots, but that only increases \opt's cost.) 
  \end{description}
  Finally, consider any phase that ends without $\brb_K$ merging all files into one
  (this must be the final phase).
  Bound (c) above holds by the same argument. 
  $\brb_K$'s cost in the phase is $\brb_{K-1}(\calI[s,s'])$
  which, by definition of $\brb_K$, since $\brb_K$ doesn't merge, is less than $(K-1)\,\ell[1,s']$.
  This and (c) imply that $\brb_K$'s cost during the phase is most $K-1$ times \opt's cost.
\end{proof}

\begin{sublemma}[Part (ii)]\label{lemma:latency worst case lower bound}
  No deterministic online algorithm for \latencyBmc is less than $K$-competitive.
\end{sublemma}
\begin{proof}
  \renewcommand{\L}[1]{L_{#1}}

  Fix any deterministic online algorithm $A$.
  We will define a \latencyBmc instance $\calI$ 
  such that $A(\calI)/\opt(\calI)$ is at least $(1+O(K/\L K))\,K$ 
  where $L_K \gg K$ is an arbitrarily large integer.
  This will prove Part (ii).
  
  The lengths in $\calI$ will be {\em well-separated},
  enabling us to use a \emph{max-based cost} in the analysis:

  \begin{subdefinition}[well-separated]
    A set of lengths is \emph{well-separated} (w.r.t.~$\L K$)
    if every two non-zero lengths in the set differ by a factor of at least $\L K$.
    Sequence $\calI$ is well-separated if its lengths are.
  \end{subdefinition}

  \begin{subdefinition}[max-based cost]
    Recall that in the definition of \latencyBmc
    merging a collection of files generates a file 
    whose length is the sum of the merged lengths.
    Modify the definition so that, instead, the merged file's length
    (and the cost of the merge) is the \emph{maximum} of the merged files's lengths.
    The \emph{max-based cost} (of a merge, or of a schedule) 
    is the cost using this modified definition.
  \end{subdefinition}

  \begin{sublemma}\label{lemma:max-based}
    For any well-separated sequence $\calI$ and any schedule $\sigma$,
    the true cost $\sigma(\calI)$
    is at most $1/(1-1/\L K) = 1+O(1/\L K)$ times its max-based cost $\sigma'(\calI)$.
  \end{sublemma}
  \begin{proof}
    With the original definition, the length of a file in the stack at any time
    is the sum $\sum_{t=i}^j \ell_t$ of some interval of lengths in the given instance $\calI$.
    With the modified definition, 
    the length of the file is instead $\max_{t=i}^h \ell_t$,
    the maximum length in the interval.
    % The sum is at least the maximum.
    Since $\calI$ is well separated,
    % the sum is at most 
    $\sum_{t=i}^j \ell_t
    \leq \max_{t=i}^j \ell_t (1+1/\L K+1/\L K^2 + \cdots) 
    = \max_{t=i}^j \ell_t /(1-1/\L K)$.
    % times the maximum.
  \end{proof}

  To prove the theorem, 
  we construct a well-separated $\calI$ 
  for which the max-based cost $\opt'(\calI)$
  is at most $1/K+O(1/\L K)$ times the true cost $A(\calI)$ of $A$ on $\calI$.

  \smallskip 

  Before we define the lengths to be used in $\calI$,
  fix $K$ integers $\L1 \gg \L2 \gg \cdots \gg \L K \gg K$,
  by choosing arbitrarily large $\L K \gg K$,
  then defining each $\L h$ for $h\in[1,K-1]$ from $\{\L {h+1},\ldots,\L K\}$ via
  \begin{equation}\label{eqn:nx1}
    \L h \, = \, \L {h+1} \L K^{N_h}
    \mbox{~~where~~} \textstyle N_h = \prod_{i=h+1}^{K} \L i.
  \end{equation}
  For each $h\in[1,K]$, 
  define the \emph{$h$-lengths}: $w_{h1} \ll  w_{h2}\ll  \cdots \ll w_{hN_h}$
  by taking $w_{hi} = \L K^{i}/\L h$.

  \smallskip                    

  \begin{sublemma}\label{lemma:separated}
    (i) The set ${\{w_{hi}\}}_{h,i}$ of lengths defined above is well separated.

    (ii) Each $h$-length $w_{hi}$ is at most $1$, but satisfies $\L h \,w_{hi} \ge \L K$.
  \end{sublemma}
  \begin{proof}
    For any $h\in[1,K]$, 
    the $h$-lengths are well-separated among themselves. 
    The largest $h$-length is $w_{h N_h}$,
    which (by~\eqref{eqn:nx1} and Def.~of $w$) 
    is at most $1/\L K$ times the smallest $(h+1)$-length $w_{h+1,1}$.
    This implies that the $h$-lengths are well-separated from the $(h+1)$-lengths,
    so the complete set is well-separated.
    It also implies that each length $w_{h,i}$ is at most $w_{K\,1}=1$.  
    By inspection, $\L h \, w_{h\,i} \ge \L K$.
  \end{proof}

  % Note that all $h$-lengths are smaller than all $(h-1)$
  % $w(h,N_h) \le w(h+1,1)/\L K$ for all $h\in[K-1]$.

  Define the request sequence $\calI$ inductively via \emph{phases}.
  A \emph{1-phase} inserts the next unused $1$-length, then repeatedly inserts zeros;
  it stops when the algorithm merges the 1-length with a larger length
  or the 1-phase has inserted $\L 1$ zeros.
  For $h\in[1,K-1]$, an \emph{$h$-phase} 
  inserts the next unused $h$-length, then repeatedly does $(h-1)$-phases;
  it stops when the algorithm merges the $h$-length with a larger length
  or the $h$-phase has done $\L h$ $(h-1)$-phases.
  A \emph{$K$-phase} reveals inserts the $K$-length $w_{k1}=1$,
  then does $\L K\,$ $(K-1)$-phases.
  The sequence $\calI$ is just a single $K$-phase.

  \smallskip

  Observe that $\calI$ uses exactly one $K$-length, 
  exactly $\L K$ $(K-1)$-lengths, 
  at most $\L K\L {K-1}$ $(K-2)$-lengths, 
  and, for $h \in [1,K]$, at most $N_h\,$ $h$-lengths (for $N_h$ from~\eqref{eqn:nx1}).

  \smallskip

  For $h\in[1,K]$, 
  let $n_h$ ($\leq N_h$) denote the total number of $h$-phases in $\calI$.
  (This depends on the algorithm.)
  For $i\in[1,n_h]$, let $n_{hi}$ denote the number of $(h-1)$-phases 
  (or number of zeros if $h=1$) within the $i$th $h$-phase.
  Note $n_K = 1$ and $n_{k1} = \L K$.

  \begin{sublemma}\label{lemma:opt}
    The max-based cost of $\opt$ on $\calI$ is at most
    \(2 + \frac{1}{K} \sum_{h=1}^K \sum_{i=1}^{n_h} w_{h\,i}\, n_{h\,i}\).
  \end{sublemma}
  \begin{proof}
    We show that there exists a schedule of at most the desired max-based cost.

    Recall that we have $K+1$ types of lengths in $\calI$:
    zeros, 1-lengths, 2-lengths, \ldots, $K$-lengths
    (in order of increasing length).
    Call zeros \emph{0-lengths}.

    % Consider first a schedule $\alpha$ for $\calI$ that uses $K+1$ slots (too many!),
    % assigning lengths to slots according to the rule
    % ``\emph{Given an $h$-length, insert it into slot $h+1$.}''
    % What is the (max-based) cost of $\alpha$ on $\calI$?
    % When an $h$-length $\ell_t$ is inserted,
    % $\alpha$ merges the length into slot $h+1$,
    % along with previously inserted $\ell$-lengths (with $\ell \le h$),
    % all of which are smaller than $\ell_t$
    % (because all $\ell$-lengths with $\ell<h$ are smaller than all $h$-lengths,
    % and $h$-lengths occur in $\calI$ in increasing order).
    % Hence, each length in $\calI$ contributes only once to the max-cost,
    % which is thus $\sum_{h=1}^K \sum_{i=1}^{n_h} w_{hi}$.
    % Further, since the lengths are well separated,
    % this sum is at most $w_{11}/(1-1/\L K)  \le 2$.
    
    Consider $K$ different $K$-slot schedules
    $\beta(1),\beta(2),\ldots, \beta(K)$,
    where, for each $b \in [1,K]$, schedule $\beta(b)$ chooses slots
    according to the following rule:
    \emph{Given an $h$-length, if $h < b$, 
      then merge it into slot $h+1$, else merge it into slot $h$.}
    That is, slot $b$ receives by $(b-1)$-lengths and $b$-lengths;
    every other length type $h$ goes in its own slot:
    $h$ (if $h<b-1$) or $h+1$ (if $h>b$). % as it did in $\alpha$.

    What is the max-cost of $\beta(b)$ on $\calI$?
    Consider the $h$-lengths $\ell_t$ with $h\ne b-1$.
    For such a length,
    $\beta(b)$ merges the length only with previously merged $\ell$-lengths
    where $\ell \le h$.
    Because all $\ell$-lengths with $\ell<h$ are smaller than all $h$-lengths,
    and $h$-lengths occur in $\calI$ in increasing order,
    these other lengths are smaller than $\ell_t$, so the max-based merge cost is $\ell_t$.
    Hence, the total cost of such merges
    is at most $\sum_t \ell_t = \sum_{h=1}^K \sum_{i=1}^{n_h} w_{hi}$.
    Further, since the lengths are well separated,
    this sum is at most $w_{11}/(1-1/\L K)  \le 2$.

    Next consider the insertion of any $(b-1)$-length $\ell_t = w_{b-1,\,j}$.
    The max-cost of its merge is the most recently revealed $b$-length, say $w_{bi}$.
    So, the $b$-length from $b$-phase $i$
    contributes its length to the aggregate max-cost
    once for each $(b-1)$-phase that occurs in $b$-phase $i$.

    In sum, the max-cost of $\beta(b)$ is at most
    $2 + \sum_{i=1}^{n_b} w_{bi} n_{bi}$.
    Hence, the max-based-costs of the $K$ schedules ${\{\beta(b)\}}_b$
    are, on average, at most the bound claimed in the lemma.
  \end{proof}

  \begin{sublemma}\label{lemma:A}
    The cost of $A$ on $\calI$ is at least
    \(
    \big(1-1/\L K\big) \sum_{h = 1}^{K} \sum_{i=1}^{n_h} n_{h\, i} \,w_{h\,i}.
    \)
  \end{sublemma}
  \begin{proof}
    When a merge occurs at time $t$,
    the cost $s^{t+1}_{\sigma_t}$ of the merge
    is the sum of some interval $\calI[i,t]$ of lengths in $\calI$;
    say each length in this interval \emph{contributes its value to the merge}.
    The total contributions of all lengths in $\calI$ (to all merges)
    equals the cost of the schedule.

    For $i\in[1,n_1]$, the $i$th $1$-phase reveals $1$-length $w_{1i}$, then $n_{1i}$ zeros.  
    Slot 1 is not emptied before the phase ends,
    so slot 1 contains $w_{1i}$ until the end of the phase,
    so each of the $n_{1i}$ zeros causes $w_{1i}$ to contribute to one merge,
    contributing in total at least $n_{1i} \,w_{1i}$.
    For $h>1$, for $i\in[1,n_h]$, the $i$th $h$-phase reveals $h$-length $w_{hi}$, 
    then does $n_{h\,i}$ $(h-1)$-phases.
    Slot $h$ is not emptied before the $h$-phase ends,
    so $w_{hi}$ is contained in a slot in $[1,h]$ until the end of the $h$-phase.
    Each $(h-1)$-phase $j$ in the $i$th $h$-phase 
    either (a) ends with a merge that empties slot $h-1$,
    which must cause $w_{hi}$ to contribute to that merge,
    or (b) \emph{times out} --- 
    that is, $(h-1)$-phase $j$ does $n_{h-1,j} = \L {h-1}$ iterations.
    Let $\tau_{hi}$ be the number of $(h-1)$-phases in the $h$-phase that time out,
    so that length $w_{hi}$'s contributions total at least $(n_{h\,i}-\tau_{h\,i}) w_{hi}$.
    Summing over the lengths,
    their total contributions sum to at least the desired lower bound,
    \(\sum_{h=1}^K\sum_{i=1}^{n_h} n_{h\,i}w_{hi}\),
    minus the \emph{timeout loss}:
    \(\sum_{h=2}^K\sum_{i=1}^{n_h}  \tau_{h\,i} w_{hi}\). 

    To bound the timeout loss by $1/\L K$ times the desired lower bound,
    we observe, for $h\ge 2$, that
    \begin{equation}\label{eqn:claim}
      \sum_{i=1}^{n_h} \tau_{h \,i}\, w_{h\,i} 
      \,\le\, \frac{1}{\L K}\sum_{j=1}^{n_{h-1}} n_{h-1,j}\, w_{h-1,j},
    \end{equation}
    because, within each $h$-phase $i$,
    each of the $\tau_{hi}$ $(h-1)$-phases that times out 
    contributes one of the $w_{hi}$'s to the left-hand sum,
    while its corresponding contribution to the right-hand sum,
    $n_{h-1,j}\, w_{h-1,j} = \L {h-1} \,w_{h-1,j}$ is, 
    by \Cref{lemma:separated} (ii), at least $\L K\, w_{hi}$.

    Summing~\eqref{eqn:claim} over $h\ge 2$,
    the timeout loss
    is at most $1/\L K$ times the desired lower bound.
  \end{proof}

  Lemmas~\ref{lemma:max-based},~\ref{lemma:opt} and~\ref{lemma:A}
  together with the observation that $w_{k1} n_{k1} = L_K$,
  imply (by algebra) that the cost of $A$
  divided by the cost of $\opt$
  is at least $(1-O(K/\L K))K$. 
  % Since $\L K$ can be taken to arbitrarily larger than $K$, this proves \Cref{thm:lb}.
\end{proof}

%%% Local Variables:
%%% mode: latex
%%% TeX-master: "main"
%%% End:

\section{Schedules for {\normalsize BMC}$_f$  as binary search trees}\label{sec:bijection}
\setNealCounters

This section proves \Cref{thm:bijection}:
for any instance $\calI$ of \bmcf,
the schedules are isomorphic to $n$-node binary search trees.
Fix any instance \calI of \bmcf.
Let $n$ be the length of \calI.

\begin{subdefinition}
  A \emph{tree for $\calI$} is any $n$-node binary search tree $T$ holding keys $\{1,2,\ldots,n\}$.

  Define $\latency(T) = \max_{t=1}^n 1+\rightDepth_T(t)$.\,\footnote 
  {The path from the root to the node with key $t$
    has $\leftDepth_T(t)$ left children and $\rightDepth_T(t)$ right children.}

  Define
  $\cost_f(T) =\sum^n_{t=1} \ell_t\, (1+\leftDepth_T(t)) + r_t\, f(1+\rightDepth_T(t))$.
\end{subdefinition}

Recall that, given a schedule $\sigma$, $k_t$ denotes the stack size that $\sigma$ yields at time $t$. 

\begin{theorem}\label{thm:bijection}
  There is a bijection $\phi$ between the schedules for \calI %(without interior merges)
  and the trees for \calI.  Further,
  for any schedule $\sigma$ and its tree $T=\phi(\sigma)$,
  for each $t\in[1,n]$, $k_t = 1+\rightDepth_T(t)$,
  and the number of times $\sigma$ merges the file inserted at time $t$ (directly or indirectly)
  is $1+\leftDepth_T(t)$.
  Hence, the bijection preserves latency and cost.  
\end{theorem}

% By \Cref{thm:bijection}, cheap schedules correspond to
% cheap binary search trees on $n$ keys, under an appropriate cost function.
% Optimal binary search trees and (equivalently) alphabetic codes have a large literature.
% For many cost functions, such trees and codes are well understood,
% and this understanding carries over directly to \bmc for some (but not all) classes of instances.
% This yields optimal and approximate schedules, faster algorithms, etc., directly for special cases, 
% and gives intuition for more general cases.

% \noindent \hrulefill 

% This section develops machinery to understand offline solutions.
% Fix any \bmcf instance \calI of length $n$, for any non-decreasing $f$. 
% WLOG, restrict to schedules that have no interior merges:
% \begin{lemma}\label{lemma:canonical}
%   There exists an optimal schedule without interior merges.
% \end{lemma}
% \begin{proof}
%   Let $\sigma$ be any optimal schedule that does interior merges.  
%   Suppose the first such merge merges files $F_i,F_{i+1},\ldots,F_j$ into one.
%   Instead, have $\sigma$ do that merge at the time $t$ when $F_j$ was at the top of the stack
%   (combining that merge with the non-interior merge done at time $t$).
%   This reduces the number of interior merges by one without increasing the cost.
% \end{proof}

Before proving \Cref{thm:bijection},
to develop intuition,
we state a natural recurrence relation for $\opt(\calI)$.  
The reader can focus on \linearBmc ($f(k)=k$).
\begin{subdefinition}
  Define $f_d(k) = f(k+d)-f(d)$ if $d\ge 1$ and $f_0 = f$.
  For each $(i,j,d)$,
  let $\calI_d[i,j]$ be the \bmcfd instance with read-cost function $f_d$
  and input sequence $(\ell_i,r_i),(\ell_{i+1},r_{i+1}),\ldots,(\ell_j,r_j)$.

  Let $\opt_d[i,j]$ denote the minimum cost of any schedule to $\calI_d[i,j]$.
  For $i>j$, let $\opt_d[i,j] = 0$.

  Let $\ell[i,j] = \sum_{h=i}^j \ell_h$, and $r[i,j] = \sum_{h=i}^j r_h$. 
\end{subdefinition}

\begin{sublemma}[recurrence relation for \bmcf]\label{lemma:recurrence}
  $\opt(\calI) =\opt_0[1,n]$ and, for $1\le i\le j\le n$ and $d\ge 0$, 
  \begin{equation}\label{eq:recurrence}
    \opt_d[i,j] = \min_{s=i\ldots j} \opt_d[i,s-1] +\ell[i,s] + r[s,j] f_d(1) + \opt_{d+1}[s+1,j].
  \end{equation}
\end{sublemma}
\begin{proof}
  Consider any schedule $\sigma$ for $\calI_0[1,n]$.
  As shown in \Cref{fig:bijection}(a),
  let $s\in[1,n]$ be the last time that $\sigma$ has stack size 1 ($k_s = 1$).
  The schedule $\sigma$ decomposes into three parts as follows:
  (i) during interval $[1,s-1]$, a schedule for $\calI_0[1,s-1]$;
  (ii) at time $s$, a merge of all files into a single file, say, $F$, at merge cost $\ell[1,s]$;
  (iii) during interval $[s+1,n]$, a schedule for $\calI_1[s+1,n]$,
  during which $F$ remains untouched at the bottom of the stack,
  so that $F$ contributes read cost $r[s,n]\,f(1)$.

  Conversely, any $s\in[1,n]$,
  schedule for $\calI_0[1,s-1]$ and schedule for $\calI_1[s+1,n]$
  yield a schedule for $\calI_0[1,n]$.
  This gives Recurrence~\eqref{eq:recurrence} for $\opt_0[1,n]$.
  The general case is similar.
\end{proof}

\begin{figure}[t]
  \vspace*{-0.4in}
  \centering
  \hfill 
  (a)\includegraphics[height=0.8in]{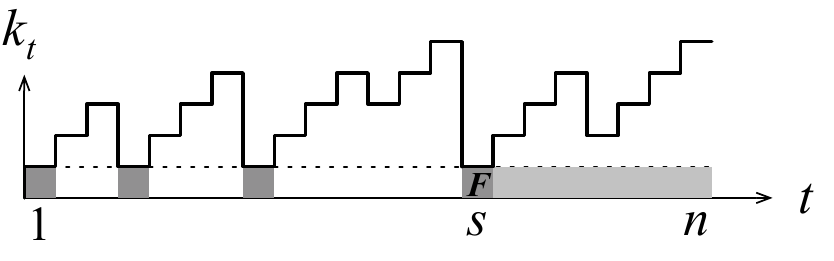}
  \hfill 
  (b)\includegraphics[height=1.2in]{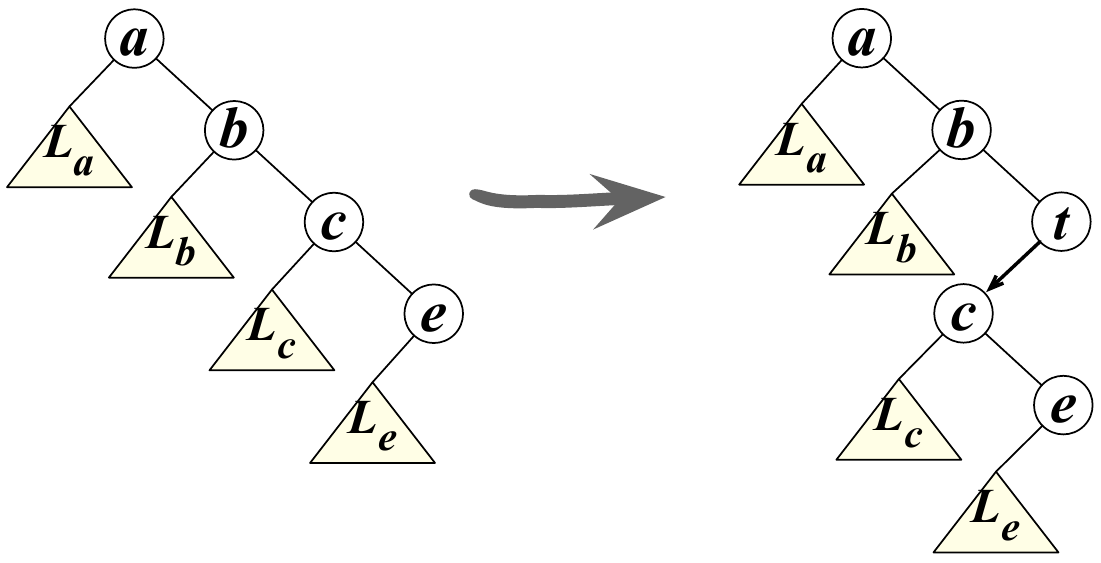}
  \hfill~{}\sf
  \caption{\sf\small
    (a) File $F$ is untouched after the last time $s$ s.t.~$k_s=1$.
    (b) Maintaining $T$ in the online setting.
  }\label{fig:bijection}
\end{figure}

% Now we prove \Cref{thm:bijection}.

\begin{proof}[Proof of \Cref{thm:bijection}.]
  Fix any schedule $\sigma$ for $\calI$. %without interior merges.
  Construct the corresponding tree $T=\phi(\sigma)$ 
  by following the inductive structure implicit in the proof of \Cref{lemma:recurrence} ---
  take $s$ to be the last time that $\sigma$ makes $k_s=1$ (see \Cref{fig:bijection}(a)),
  make $s$ the key of the root,
  then recurse on intervals $[1,s-1]$ and $[s+1,n]$, respectively, to build $T$'s left and right subtrees.
  An easy inductive argument shows that every node has the desired left and right depth. 
  Given any tree $T$ for $\calI$, the construction can be inverted 
  to construct a corresponding schedule $\sigma$, completing the proof.
\end{proof}

\begin{corollary}\label{cor:alg}
  There is an $O(n^4)$-time dynamic-programming algorithm for offline \bmcf.
  For \latencyBmc and \linearBmc, the time reduces to
  $O(K n^3)$ and $O(n^3)$, respectively.
\end{corollary}

\paragraph{Online \bmcf is equivalent to building a binary search tree online.}
Via \Cref{thm:bijection}, online \bmcf has a natural interpretation as the following online problem.
Given a \bmcf instance $\calI$, as each pair $(\ell_t, r_t)$ is revealed,
the algorithm $A$ must maintain a tree $T$ for $\calI[1,t]$.
At time $t=1$, the tree $T$ is a single node with key 1.
At each time $t>1$, $A$ must insert a new node with key $t$ into $T$,
without changing the relations of nodes already in $T$.
That is, $A$ either appends the new node to the right spine (as the right child of the bottom node), 
or inserts the new node into the right spine above some node $c$,
moving $c$ to the left child of the new node (the new node has no right child),
as shown in \Cref{fig:bijection}(b).
The goal is to minimize $\cost_f(T)$.

By a straightforward induction, 
valid sequences of insertions correspond to valid sequences of compactions.
The current tree $T$ at time $t$ corresponds (via \Cref{thm:bijection})
to the schedule of compactions over $[1,t]$.
The nodes along the right spine of $T$ correspond to the files in the stack at time $t$.
We summarize this as follows:

\begin{sublemma}\label{lemma:online tree}
  The $c$-competitive online algorithms for the problem above
  correspond to the $c$-competitive online algorithms for \bmcf.
\end{sublemma}
 
%%% Local Variables:
%%% mode: latex
%%% TeX-master: "main"
%%% End:

\section{Worst-case analysis of linear BMC}
\label{sec:linear worst case}
\setNealCounters

\begin{definition}
  In any tree $T$ for $\calI$,
  let $T_t$, $L_t$, and $R_t$ denote, respectively, the subtree with root key $t$
  and its left and right subtrees.
  In any subtree $T_t$, the keys in $T_t$ form an interval $[i,j]$.
  Let $\ell[T_t] = \ell[i,j] = \sum_{t=i}^j \ell_t$ and $r[T_t] = r[i,j] = \sum_{t=i}^j r_t$. 
  (Define $\ell[L_t] = r[R_t] = 0$ for empty $L_t$, $R_t$.)
\end{definition}

\begin{lemma}[lower bound on \opt for \linearBmc]\label{thm:lower bound}\label{lemma:lower bound}
  For any instance \calI of \linearBmc,
  any schedule $\sigma$, and its tree $T=\phi(\sigma)$,
  \begin{enumerate}\itemsep0in
  \item[(i)] \(\cost(T) = \sum_{t=1}^n \ell_t + r_t + \ell[L_t] + r[R_t]\), and
  \item[(ii)] \(\opt(\calI) 
    \ge \cost(T) -\sum_{t=1}^n\max\{\ell[L_t], \,r[R_t]\} 
    = \sum_{t=1}^n \ell_t + r_t + \min\{\ell[L_t], \,r[R_t]\}\).
  \end{enumerate}
\end{lemma}
\begin{proof}
  Part (i) follows by calculation from the definition of $\cost(T)$.
  To prove Part (ii), let $T^*$ be a tree of cost $\opt(\calI)$.
  Transform $T^*$ into $T$, without increasing the cost by much, as follows.
  Let $s$ be the root of $T$. 
  First transform $T^*$ into a tree $T'$ with $s$ at the root.
  \begin{figure}\sf
    %\vspace*{-0.4in}
    \centering
    \includegraphics[height=1in]{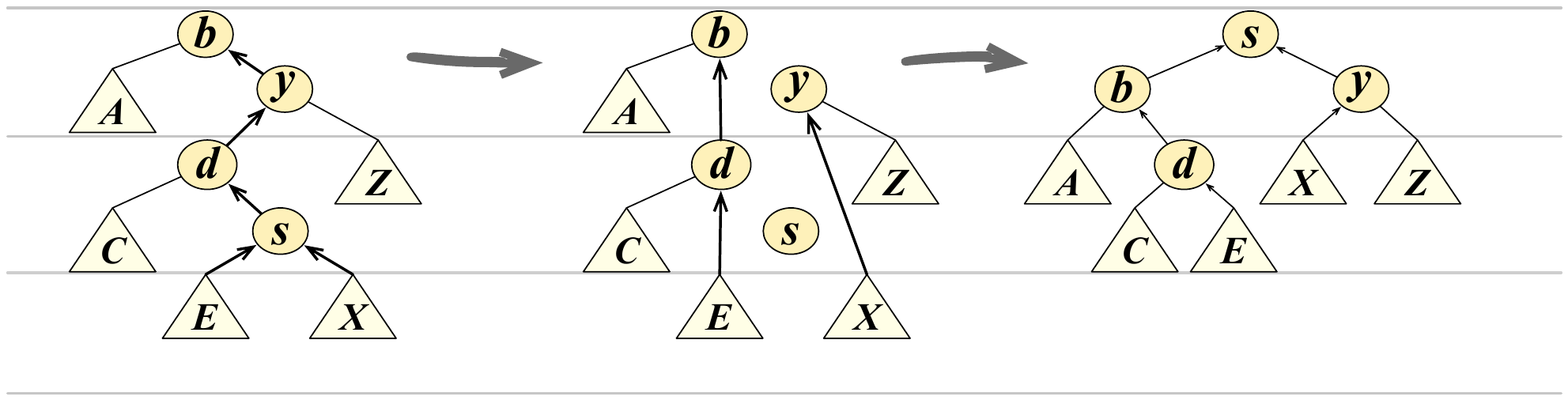}
    \vspace*{-0.1in}
    \caption{\sf
      In the proof of \Cref{thm:lower bound},
      moving $s$ to the root to transform $T^*$ into $T'$.
    }\label{fig:arman lower bound}
  \end{figure}
  In $T^*$, 
  for each node $x < s$,
  change the parent to the first ancestor less than $s$ (if any).
  For each node $x > s$,
  change the parent to the first ancestor greater than $s$ (if any).
  This splits $T^*\setminus\{s\}$ into a tree $T'_<$ for $[1,s-1]$
  and a tree $T'_>$ for $[s+1,n]$,
  as shown in \Cref{fig:arman lower bound}.
  Make $s$ the root of $T'$, with $T'_<$ as the left subtree and $T'_>$ as the right subtree.
  This defines $T'$.
  To complete the transformation, transform the left and right subtrees of $T'$
  recursively into, respectively, the left and right subtrees of $T$.

  % Next show that $\cost(T') \le \cost(T^*) + \max\big\{\ell[L_t], \, r[R_t] \,f_{d(t)}(1)\big\}$. 
  How are left and right depths of nodes changed 
  in the transformation from $T^*$ to $T'$?
  If the root of $T^*$ is smaller than $s$
  (as in \Cref{fig:arman lower bound})
  then the only depths that may increase
  are the left depths of nodes in the left subtree of $T'$,
  which increase by at most 1.
  Hence, $\cost(T') \le \cost(T^*) + \ell[L_t]$. 
  Similarly, if the root of $T^*$ is larger than $s$, then $\cost(T') \le \cost(T^*) + r[R_t]$. 
  It follows that $\cost(T') \le \cost(T^*) + \max\{\ell[L_s], \, r[R_s]\}$. 
  
  By induction, transforming $T'$ into $T$ by recursing into $T'$s two subtrees
  increases the cost by at most $\sum_{t\ne s} \max\{\ell[L_t], \, r[R_t]\}$,
  so the total cost increase in transforming $T^*$ into $T$ is at most
  $\sum_{t=1}^n\max\{\ell[L_t], \, r[R_t]\}$. 
  It follows that
  $\opt(\calI) = \cost(T^*) 
  \ge \cost(T) - \sum_{t=1}^n\max\{\ell[L_t], \, r[R_t]\}$.
\end{proof}

For intuition, note that \Cref{thm:lower bound} gives
a fast offline 2-approximation algorithm:
\begin{corollary}\label{cor:2-approx}
  There is an $O(n)$-time, offline 2-approximation algorithm for \linearBmc.
\end{corollary}
\begin{proof}
  Fix an instance $\calI$, schedule $\sigma$ and its tree $T$.
  Say node $t$ in $T$ is \emph{balanced} if
  $|\ell[L_t] - r[R_t]| \le \ell_t + r_t$.

  By \Cref{thm:lower bound}(i),
  \(
  \cost(T) 
  = \sum_{t=1}^n~ \ell_t + r_t ~+~ \ell[L_t] + r[R_t]\).
  Comparing this sum term-by-term with the lower bound on $\opt(T)$ 
  from \Cref{thm:lower bound}(ii),
  it follows that if every node in $T$ is balanced, 
  then $\cost(T)\le 2\opt(\calI)$:
  \[\cost(T) 
  = \sum_{t=1}^n~ \ell_t + r_t 
  ~+~ 2\min\{\ell[L_t], \,r[R_t]\}
  ~+~ |\ell[L_t] - r[R_t]|
  \le 
  \sum_{t=1}^n~ 2(\ell_t + r_t ~+~ \min\{\ell[L_t], \,r[R_t]\})\]
  To construct such a $T$,
  use binary search to find the maximum $s\in[1,n+1]$ 
  such that $\ell[1,s-1] \le r[s,n]$ (so $\ell[1,s]>r[s+1,n]$, this ensures $s$ is balanced).
  Make $s$ the root of $T$, then recurse on $[1,s-1]$ and $[s+1,n]$.
\end{proof}

We note without proof that \Cref{thm:lower bound} and \Cref{cor:2-approx} 
extend to \bmcf for any concave $f$.\footnote
{Define $f_d(k) = f(k) - f(d)$ if $d>0$, and $f_0 = f$.
Let $d(t)=\rightDepth_T(t)$.
Then 
(i) \(\cost_f(T) = \sum_{t=1}^n \ell_t +  \ell[L_t] + (r_t + r[R_t])\, f_{d(t)}(1)\)
and
(ii) \(\opt(\calI) \ge \cost_f(T) - 
\sum_{t=1}^n~ \ell_t + r_t\, f_{d(t)}(1) ~+~ 
\min\{\ell[L_t], \, r[R_t] \,f_{d(t)}(1)\}\).}

\smallskip

Next we develop the online algorithm $A$. 
We describe $A$ as an online algorithm for maintaining a tree $T$,
per \Cref{lemma:online tree}.
To guarantee $\lambda$-competitiveness, we ensure that
$\cost_f(T)$ is at most $\lambda$ times the lower bound $T$ gives via \Cref{thm:lower bound}.
$A$ maintains the following invariant on $T$:
\begin{equation}\label{eq:invariant}
  \forall s\in [1,t].~~ \ell[L_s] \ge r[R_s].
\end{equation}
At each time $t$, $A$ inserts the new node with key $t$
\emph{as high as possible on the right spine, subject to Invariant~\eqref{eq:invariant}.}
(Inserting $t$ at the bottom of the spine is one way to maintain the invariant.)

\begin{theorem}[\linearBmc worst-case analysis]\label{thm:linear worst case}
  The online algorithm $A$ above
  is $O(1)$-competitive on those instances $\calI$ of \linearBmc 
  such that $\ell_t = O(r_t)$ for all $t$.
  % or for some $\alpha>0$, $\ell_t = \alpha\, r_t$ for all $t$.
\end{theorem}
\begin{proof}
  Fix any instance $\calI$ such that $\ell_t \le \alpha\,r_t$ for all $t$ (where $1\le \alpha = O(1)$).
  We use an amortized analysis to show that $\cost(T)$ 
  is always at most $1+\alpha$ times the lower bound 
  that $T$ gives on $\opt$ via \Cref{thm:lower bound}(ii).
  Let ${\cal S}(T)$ denote the nodes in $T$ that are on the right spine.  

  As $A$ maintains $T$, define the \emph{potential} of $T$ to be
  \begin{equation}\label{eq:potential}
    \Phi(T) = \sum_{x\in T} \big(\ell_x + r_x + r[R_x]\big)\times
    \begin{cases} 
      1 & x\,\in {\cal S}(T) \\
      2 & x\,\not\in {\cal S}(T).
    \end{cases}
  \end{equation}
  By inspection of $\Phi$,
  Invariant~\eqref{eq:invariant} 
  implies that $\Phi(T)$ is $O(1)$ times the lower bound from \Cref{thm:lower bound}.
  By calculation, at time step $t$, the increase in $\cost(T)$ 
  is $k_t\, r_t + \ell_t + \ell[T_c]$, where $k_t$ is the number of nodes on the right spine after time $t$
  and $c$ is the node that becomes the left child of $t$ after the insertion.
  (as in \Cref{fig:bijection}(b)).
  To finish, we verify by calculation (using $\ell[R_c]\le \alpha\, r[T_c]$)
  that this increase is less than $1+\alpha$ times the increase in $\Phi(T)$.
  That is, $\Delta\!\cost(T) \le (1+\alpha)\,\Delta \Phi(T)$. 
  
  Consider the insertion of node $t$.
  Recall $\cost(T) = \sum_{x\in T} \ell_x + r_x + \ell[L_x] + r[R_x]$.
  First consider the case 
  when $t$ is inserted at the bottom of the right spine.
  Then $\cost(T)$ increases by $\ell_t + k_t\, r_t$.
  The potential increases by $\ell_t +  (k_t+1)\, r_t$, so we are done.
  Otherwise, $t$ is inserted along the right spine,
  with node $c$ on the spine becoming the left child of $t$.
  Let $k_t$ be the length of the spine after the insertion.
  Now,
  \begin{align}
    \Delta\Phi(T) &\ge k_t\, r_t + \ell_t + \ell_c + r[R_c]
    &&\text{Inspecting $\Phi$, using that $c$ leaves spine ${\cal S}(T)$.}
       \label{eq:pot increase}
    \\
    \Delta\!\cost(T) &= k_t\, r_t + \ell_t + \ell[T_c] 
    && \text{Using $L_t = T_c$ and $R_t = \emptyset$ and def'n of $\cost$.}
       \label{eq:cost increase}
    \\
    \ell[T_c] & = \ell_c + \ell[L_c] + \ell[R_c] 
    && \text{By definition of $\ell[X]$.} 
       \label{eq:Tc}
    \\
    \ell[L_c] & < r_t + r[R_c]
    && \text{By the algorithm's choice of } c. 
       \label{eq:Lc}
    \\
    \ell[R_c] & \le \alpha\, r[R_c] 
    && \text{By the assumption } \forall x.~\ell_x \le \alpha r_x.  
       \label{eq:Rc}
    \\
    \Delta\!\cost(T) &< k_t\, r_t + \ell_t + \ell_c + r_t + (1+\alpha)\, r[R_c]
    && \text{Transitively from~\eqref{eq:cost increase}--\eqref{eq:Rc}.}
       \label{eq:cost increase 2}
    \\
    \Delta\!\cost(T) & < (1+\alpha)\,\Delta \Phi(T)
    && \text{Comparing~\eqref{eq:pot increase} and~\eqref{eq:cost increase 2}.} \notag
  \end{align}
\end{proof}

%%% Local Variables:
%%% mode: latex
%%% TeX-master: "main"
%%% End:

\section{Average-case analyses of BMC$_K$ and linear BMC}
\label{sec:average case}
\setNealCounters

\begin{theorem}\label{thm:average case}
  \LatencyBmc and \linearBmc have online algorithms $A$ and $B$, respectively, 
  that are asymptotically 1-competitive in expectation on random inputs $\calI$ 
  with bounded, i.i.d.~requests.
  Let $\calI=\langle (\ell_t, r_t) \rangle_t$ be a random sequence of $n$ i.i.d.~pairs from 
  any bounded probability distribution over $\Rp\!\times\Rp$.  
  Let $(\lbar, \rbar) = (\E[\ell_t],\E[r_t])$ for all $t$.
  For \latencyBmc,
  \[\E_{\calI}[A(\calI)]\,\sim\, \E_{\calI}[\opt(\calI)]\, \sim\, \lbar\, K n^{1+1/K}/c_K\]
  where $c_K = (K+1)/(K!)^{1/K}$ (so $c_K\rightarrow e$ for large $K$). 
  \[\E_{\calI}[B(\calI)]   \,\sim\, \E_{\calI}[\opt(\calI)]\,\sim\, \beta\, n \log_2 n,~~~~~~\,\]
  where $\beta$ satisfies $1/2^{\lbar/\beta} + 1/2^{\,\rbar/\beta}$,
  so $\beta = \Theta(\lbar + \rbar)/\ln(1+ \max(\lbar/\rbar,\, \rbar/\lbar))$.
\end{theorem}

We conjecture that \brb is also asymptotically 1-competitive on bounded i.i.d.~inputs.

\smallskip

Before we prove the theorem, we prove two utility lemmas.
The first characterizes optimal costs on \emph{uniform} instances $\overline\calI$,
that is, $\overline\calI=(\lbar,\rbar)^n$ 
for some $(\lbar,\rbar)\in \Rp\!\times\Rp$:

\begin{sublemma}[uniform instances]\label{lemma:opt uniform}
  Fix any $(\lbar,\rbar) \in \Rp\!\times\Rp$. 
  Let $\overline\calI = (\lbar,\rbar)^n$.
  \begin{enumerate}\itemsep0in
  \item[(i)] For \latencyBmc, $\opt(\overline\calI) \sim \lbar\,K n^{1+1/K}/c_K$,
    for $c_K$ as defined in \Cref{thm:average case}.
  \item[(ii)]
    For \linearBmc, $\opt(\overline\calI) \sim \beta\, n\log n$,
    for $\beta$ such that~$1/2^{\lbar/\beta} + 1/2^{\rbar/\beta} = 1$.
  \end{enumerate}
\end{sublemma}
The value of $\beta$ is $\Theta(\max\{\lbar/\log \,\lbar/\rbar,\, \rbar/\log\,\rbar/\lbar\})$. 
\begin{proof}
  By \Cref{thm:bijection}, the optimal costs equal the costs
  of optimal $n$-node binary search trees under an appropriate cost function.
  For uniform instances, 
  these cost functions are well-studied,
  and optimal costs are known 
  to asymptotically equal these quantities
  (e.g.~\cite{bentley_general_1982,golin_more_2008,kapoor_optimum_1989}).
  Here are the details.

  \noindent\emph{(i)}
  For the read-cost function $f$ for \latencyBmc,
  the tree $T$ for $\calI$ that minimizes $\cost_f(T)$ has right-depth at most $K-1$,
  and, subject to that constraint, has $n$ nodes chosen to minimize total left-depth.
  This $T$ is well understood (e.g.~\cite{bentley_general_1982}).
  % They can have up to $i+j\choose i$ nodes of left-depth exactly $i$ and right-depth exactly $j$
  % (assuming $j\le K-1$),
  % and up to $K+ d+ 1 \choose K$ nodes of left-depth at most $d$ and right-depth at most $K-1$.
  $T$ has maximum left-depth $d$,
  where, by calculation, $d$ is minimum subject to ${K+ d \choose K} \ge n$, so $d \sim (K! n)^{1/K}$.
  $T$ has total left-depth $\sim \frac{K}{K+1}d n$.
  % (Further, for each $i\in[0,d-1]$ and $j\in[0,K-1]$,
  % $T$ has exactly $i+j\choose i$ nodes of left-depth $i$ and right-depth $j$.)
  By \Cref{thm:bijection}, $\opt(\calI) \sim \frac{K}{K+1}d n = K n^{1+1/K} / c_K$.
  \smallskip

  \noindent\emph{(ii)}
  For the read-cost function $f$ for \linearBmc, 
  the tree for $\overline\calI$ that minimizes $\cost_f(T)$
  corresponds to an optimal \emph{lopsided alphabetic code}
  --- a sequence of $n$ distinct (and ordered) binary codewords $C_1,C_2,\ldots,C_n$,
  where the cost of $C_t$ is
  $\lbar$ times the number of zeros in $C_t$ plus $\rbar$ times the number of ones.
  Such codes are well-studied (e.g.,~\cite{golin_more_2008,kapoor_optimum_1989}),
  and have minimum total cost $\sim \beta\, n\log n$.
  By \Cref{thm:bijection}, $\opt(\calI) \sim\beta\, n\log n$.
\end{proof}

As an aside, this approach extends to other special cases.
For example, consider any ``proportional'' instance $\calI$ of \linearBmc
such that, for some $\alpha>0$, each pair $(\ell_t, r_t)$ satisfies $\ell_t = \alpha\, r_t$.
Then $\opt(\calI) \sim \beta\, r[1,n]\,H(p)$,
where $H(p)$ is the entropy of the distribution $p$ such that $p_t = r_t/r[1,n]$, 
and $\beta$ is such that $1/2^{\alpha/\beta} + 1/2^{1/\beta} = 1$~\cite{golin_more_2008}.

\smallskip

Next we prove that one can replace uniform requests by bounded, i.i.d.~requests
without changing optimal asymptotic costs.
For the remainder of the proof,
let $\calI$, $\lbar$, and $\rbar$ be as in \Cref{thm:average case}.
Let $\overline\calI = \E[\calI] = (\lbar, \rbar)^n$.
Take $\delta = 100\, U \log(n)/\,n\eps^2$, 
where $U \ge \max_t \max (\ell_t/\lbar, r_t/\rbar)$ 
gives an absolute upper bound on lengths and read costs 
from the distribution,
and $\eps \rightarrow 0$ slowly as $n\rightarrow \infty$
(e.g.~$\eps=1/\log n$), so $\eps=o(1)$.
Call intervals $[i,j]$ of length at least $\delta n$ \emph{large}, and the rest \emph{small}.
Say that $\calI$ \emph{behaves} if
$\ell[i,s] + r[s,j] \ge (1-\eps)[(s-i+1)\lbar+(j-s+1)\rbar]$
and $\ell[i,j] \ge (1-\eps)(j-i+1)\lbar$
for every large interval $[i,j]\subseteq[1,n]$ and every $s\in[i,j]$.

\begin{sublemma}\label{lemma:behaves}
  $\calI$ behaves with probability $1-o(n^{-10})$.
\end{sublemma}
\begin{proof}
  This follows from a standard Chernoff bound and the naive union bound. 
  Here are the details.
  Consider any large $[i,j]$ and $s\in[i,j]$.
  By a standard Chernoff bound, using $j-i+1\ge \delta n$,
  \[\Pr[\ell[i,j] \le (1-\eps) (j-i+1)] 
  ~\le~ 
  \exp(-\eps^2 (j-i+1)\lbar / (3 U/\lbar)) 
  ~\le~ 
  \exp(-33 \log n ) 
  ~=~
  n^{-33}.
  \]
  Likewise, \(\Pr[\ell[i,s] + r[s,j] \le (1-\eps) [(s-i+1)\lbar + (j-s+1)\rbar]\)
  is at most $n^{-33}$.
  Since there are at most $n^3$ triples $(i,s,j)$,
  the probability that $\calI$ misbehaves is at most $2n^{-30} = o(n^{-10})$.
\end{proof}

\begin{sublemma}\label{lemma:opt iid}
  For both \latencyBmc and \linearBmc,
  $\E_{\calI}[\opt(\calI)] \sim \opt(\overline\calI)$.
\end{sublemma}
\begin{proof}
  Let $\overline\sigma$ be an optimal schedule for $\overline\calI$.
  %Recall that $\overline\sigma(\calI)$ is the cost of using schedule $\overline\sigma$ to handle $\calI$.  
  Then 
  $\E[\opt(\calI)] 
  \le \E[\overline\sigma(\calI)] 
  = \overline\sigma(\E[\calI])
  = \overline\sigma(\overline\calI)
  = \opt(\overline\calI)$.
  (The first equality holds by linearity of expectation,
  as $\overline\sigma(\calI)$ is a linear function of $\calI=\langle(\ell_t,r_t)\rangle_t$.)
  This shows $\E[\opt(\calI)] \le \opt(\overline\calI)$.
  It remains to show
  $\E_{\calI}[\opt(\calI)] \ge (1-o(1)) \opt(\overline\calI)$.

  First we prove the claim for \linearBmc.
  For \linearBmc Recurrence~\eqref{eq:recurrence} simplifies to
  \begin{equation}\label{eq:recurrence linearBmc}
    \opt[i,j] = \min_{s=i\ldots j} \opt[i,s-1] +\ell[i,s] + r[s,j] + \opt[s+1,j].
  \end{equation}
  Assume that $\calI$ behaves.  Then (by induction on the recurrences)
  $\opt[i,j] \ge (1-\eps)\lb[i,j]$, where
  \begin{equation}\label{eq:recurrence lb}
    \lb[i,j] = \min_{s=i\ldots j} \lb[i,s-1] + (s-i+1)\, \lbar + (j-s+1)\,\rbar + \lb[s+1,j]
  \end{equation}
  for large intervals $[i,j]$ and $\lb[i,j] = 0$ for small $[i,j]$.
  To finish we show $\lb[1,n] \ge (1-o(1))\opt(\overline \sigma)$.
  Let $T$ be the recursion tree for Recurrence~\eqref{eq:recurrence lb} for $\lb[1,n]$,
  interpreted as a binary search tree on keys $[1,n]$ as in the proof of Thm.~\ref{thm:bijection}.
  In $T$, for each maximal subtree $S$ whose interval $[i,j]$ is small,
  replace $S$ by the optimal subtree for $\overline\calI[i,j]$.
  Let $T'$ be the resulting tree.
  Using $T'$ as a solution (schedule) for $\opt(\overline\calI)$,
  and letting $S$ range over the subtrees introduced into $T'$,
  \(
  \opt(\overline\calI) \le \cost(T') = \lb[1,n] + \sum_{S} \cost(S). 
  \)

  The number of subtrees $S$ is at most $n/\delta n = 1/\delta$.
  Each has $\cost(S) = O(\beta \,\delta\, n \log(\delta n))$ (\Cref{thm:average case}(ii)),
  so
  $\sum_{S} \cost(S)$ is $O((1/\delta) (\beta\,\delta\, n\log(\delta n)))$,
  which is $o(\opt(\overline\calI))$, as $\delta n = \log^{O(1)} n$.

  Hence
  $\E_\calI[\opt(\calI)]\ge \Pr[\calI \text{ behaves}] (1-o(1)) \opt(\overline\calI) \sim \opt(\overline\calI)$.

  To finish, we prove the claim for \latencyBmc.
  We show 
  $\E_{\calI}[\opt(\calI)] \ge (1-o(1)) \opt(\overline\calI)$ for \latencyBmc. 
  The idea is the same as for \linearBmc.
  % For \latencyBmc, Recurrence~\eqref{eq:recurrence} simplifies to
  % \[
  % \opt_d[i,j] = 
  % \min_{s=i\ldots j} \opt_d[i,s-1] + \ell[i,s] + \opt_{d+1}[s+1,j]
  % \]
  % for $d< K$ and $i\le j$,
  % while $\opt_{K}[i,j] = \infty$ for $i\le j$
  % and $\opt_d[i,j]=0$ for $i>j$.
  Define $\lb_{0}[1,n]$ by recurrence
  \[
  \lb_d[i,j] = 
  \min_{s=i\ldots j} \lb_d[i,s-1] + \lb_{d+1}[s+1,j]
  + \begin{cases}
    \ell[i,s]& \text{if $[i,s]$ large}, \\
    0 & \text{otherwise},
  \end{cases}
  \]
  for $d< K$ and $[i,j]$ large,
  while $\lb_{K}[i,j] = \infty$ for $i\le j$,
  and otherwise $\lb_d[i,j]=0$ for $[i,j]$ small.
  As in the proof sketch,
  if $\calI$ behaves, then $\opt(\calI) \ge (1-\eps)\lb_0[1,n]$.
  Let $T$ be the recurrence tree for $\lb_0[1,n]$.
  Interpret $T$ as a solution for $\overline\calI$,
  and, for each maximal subtree $S$ for a subproblem $\overline\calI_d[i,j]$ 
  where $[i,j]$ is small, replace $S$ by the optimal subtree for $\calI_d[i,j]$.
  Call the resulting tree $T'$.
  Then, interpreting $T'$ as a solution for $\overline\calI$,
  and letting $S$ range over the subtrees introduced into $T'$,
  $\opt(\overline\calI) \le \cost(T') \le \lb_0[1,n] + 2\sum_S \cost(S)$.
  
  (The factor of 2 accounts for each term $\ell[i,s]$ that can be ``missing''
  for the parent of each subtree $S$, in the recurrence for $\lb_d[i,j]$.)
  There are at most $n/\delta n = 1/\delta$ subtrees $S$,
  each with $\cost(S) = O((\delta n)^2)$, so
  $\sum_S \cost(S)$ is $O(\delta n^2) = O(n\log^{O(1)} n) = o(\opt(\overline\calI))$.
\end{proof}

Finally we prove \Cref{thm:average case}.

\begin{proof}
  First consider the case when $n$ and the distribution $p$ are known.
  On input $\calI$,
  have $A$ ignore the input, and do merges exactly as $\opt(\overline\calI)$ would.
  Then as a function of the input vector $\calI$, 
  the function $\calI\mapsto A(\calI)$ is linear.
  By linearity of expectation, 
  $\E[A(\calI)] = A(\overline \calI) = \opt(\overline \calI)$,
  which asymptotically equals $\E[\opt(\calI)]$ by \Cref{lemma:opt iid}.

  To handle the case when $p$ and $n$ are not known,
  use the fact that the optimal schedule for $\overline\calI$
  depends only on two parameters: $\lbar$ and $\rbar$.
  At each time $t$ that is a power of two, start a new \emph{phase}:
  merge all files into one file $F$,
  then, during the phase $[t,2t-1]$ ignore $F$ completely and
  follow the optimal schedule for $(\lbar',\rbar')^t$,
  where $\lbar'$ and $\rbar'$ are the average file length and read rate so far.

  The total cost for the merges at the start of each phase 
  and for the bottom stack slot is $O(\ell[1,n] + r[1,n]) = o(\opt(\overline\calI))$.
  We bound the remaining cost.
  Take $\delta$, $\eps$, and $U$ as earlier defined.
  The cumulative cost of the online algorithm
  through the phase containing time $\delta n$ 
  is $O(U\lbar(\delta n)^2) = o(\opt(\overline \calI))$
  (using $\delta n = O(\log n)$ and $\opt(\overline\calI) = \Omega(n\log n)$).
  After that time, with high probability,
  the estimates of $\lbar$ and $\rbar$ are all $(1\pm\eps)$-accurate,
  so, phase by phase, the expected cost of the online algorithm 
  tracks the cost of $\opt((\lbar,\rbar)^t)$ within a $1+o(1)$ factor.
  (To handle phase $[t,2t-1]$,
  the algorithm follows a static schedule, say $\sigma$, for $\opt((\lbar',\rbar')^{t})$,
  and incurs expected cost
  $\sigma((\lbar,\rbar)^{t})
  \,\le\, (1+\eps) \sigma((\lbar',\rbar')^{t})
  \,\le\, (1+\eps)\opt((\lbar',\rbar')^{t})
  \,\le\, (1+\eps)^2\opt((\lbar,\rbar)^{t})$.)
  Hence, the expected cost of the algorithm after the phase containing time $\delta n$
  is $(1+o(1))\opt(\overline\calI) = (1+o(1))\E[\opt(\calI)]$.
\end{proof}

%%% Local Variables:
%%% mode: latex
%%% TeX-master: "main"
%%% End:

\section{Benchmarks}\label{sec:benchmarks}

For \latencyBmc, we test \brb and Google's {\sc Default} algorithm
(\emph{merge minimally, subject to the constraint that each file 
remains as large as all files above it combined}).
For \linearBmc we test the algorithms from \Cref{thm:linear worst case}
and \Cref{thm:average case}.
The inputs are sequences 
with read costs i.i.d.~from an exponential distribution
and file lengths i.i.d.~from a log-normal distribution.
We let $\mu$ and $v$ denote the mean and variance of the underlying normal distribution.
When computationally feasible, we also test \opt.
Each plot plots average cost \emph{per time step} (that is, total cost divided by $n$) versus $n$,
for several algorithms on one input.

\paragraph{Results for \latencyBmc.}
Recall that for \latencyBmc, we expect \opt to cost about $\lbar\, K n^{1/K}/e$ (per time step).
We hope that \brb costs about the same.
On uniform instances, by calculation {\sc Default} costs about $\lbar\, n /(2\cdot 3^{K-1})$ per time step.
We expect {\sc Default} to have roughly this cost on i.i.d.~instances as well.
As a consequence, we expect that \brb should substantially outperform {\sc Default}
for large $n$, say, for $n \ge K 3^K$. We do see this.  
We also see that, in general, \brb is close to \opt, and better than {\sc Default} even for small $n$.
See Fig.~\ref{fig:p1} for an example.

\newcommand{\tmpH}{2.6in}
\newcommand{\tmpG}{2.6in}

\begin{figure}
  \vspace*{-.1in}
  \begin{subfigure}{.5\textwidth}
    \centering
    \includegraphics[width=2.7in,height=\tmpH]{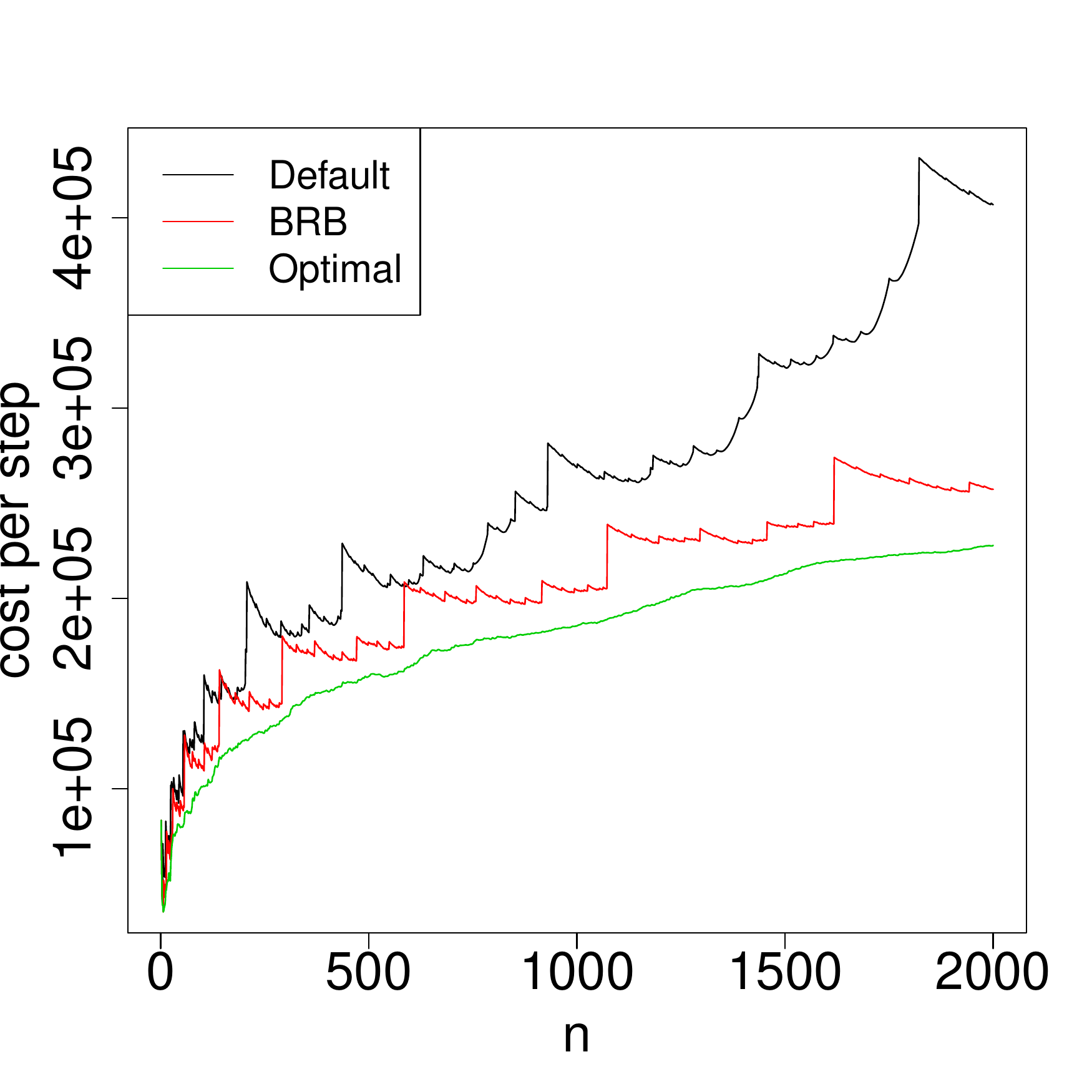}
    \caption{\LatencyBmc with $K = 5, n\le 2000$}\label{fig:p1 left}
  \end{subfigure}%
  \begin{subfigure}{.5\textwidth}
    \centering
    \includegraphics[width=2.7in,height=\tmpH]{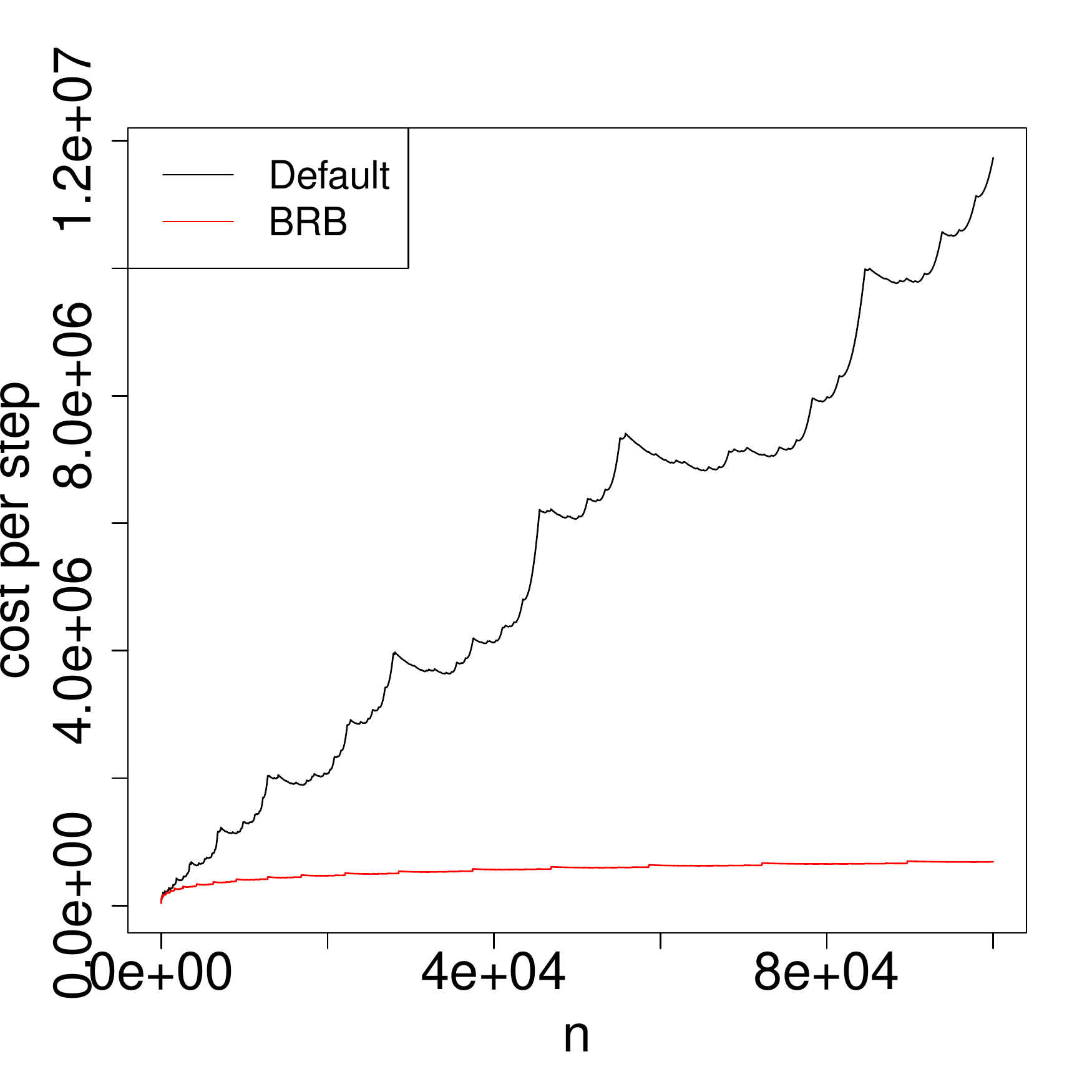}
    \caption{\LatencyBmc with $K = 5, n\le 100,000$}\label{fig:p1 right}
  \end{subfigure}
  \caption{An instance with $\mu = 10$, $v=1$, so typically $\ell_t\in[e^9,e^{11}]$.}\label{fig:p1}
\end{figure}

\paragraph{Results for \linearBmc.}
Recall that for \linearBmc, we expect \opt to cost about $\beta \log n$ (per time step),
where $1/2^{\lbar/\beta} + 1/2^{\rbar/\beta} = 1$.
We hope that our online algorithms achieve cost near this.
(We know that the \linearBmc algorithm from \Cref{thm:average case} does \emph{asymptotically}.)
We find that they do, even for small $n$, except that when $\lbar/\rbar$ is large,
the algorithm from \Cref{thm:linear worst case} doesn't do as well.
See Fig.~\ref{fig:p2} for an example.

\begin{figure}
  \vspace*{-.1in}
  \begin{subfigure}{.5\textwidth}
    \centering
    \includegraphics[width=2.7in,height=\tmpG]{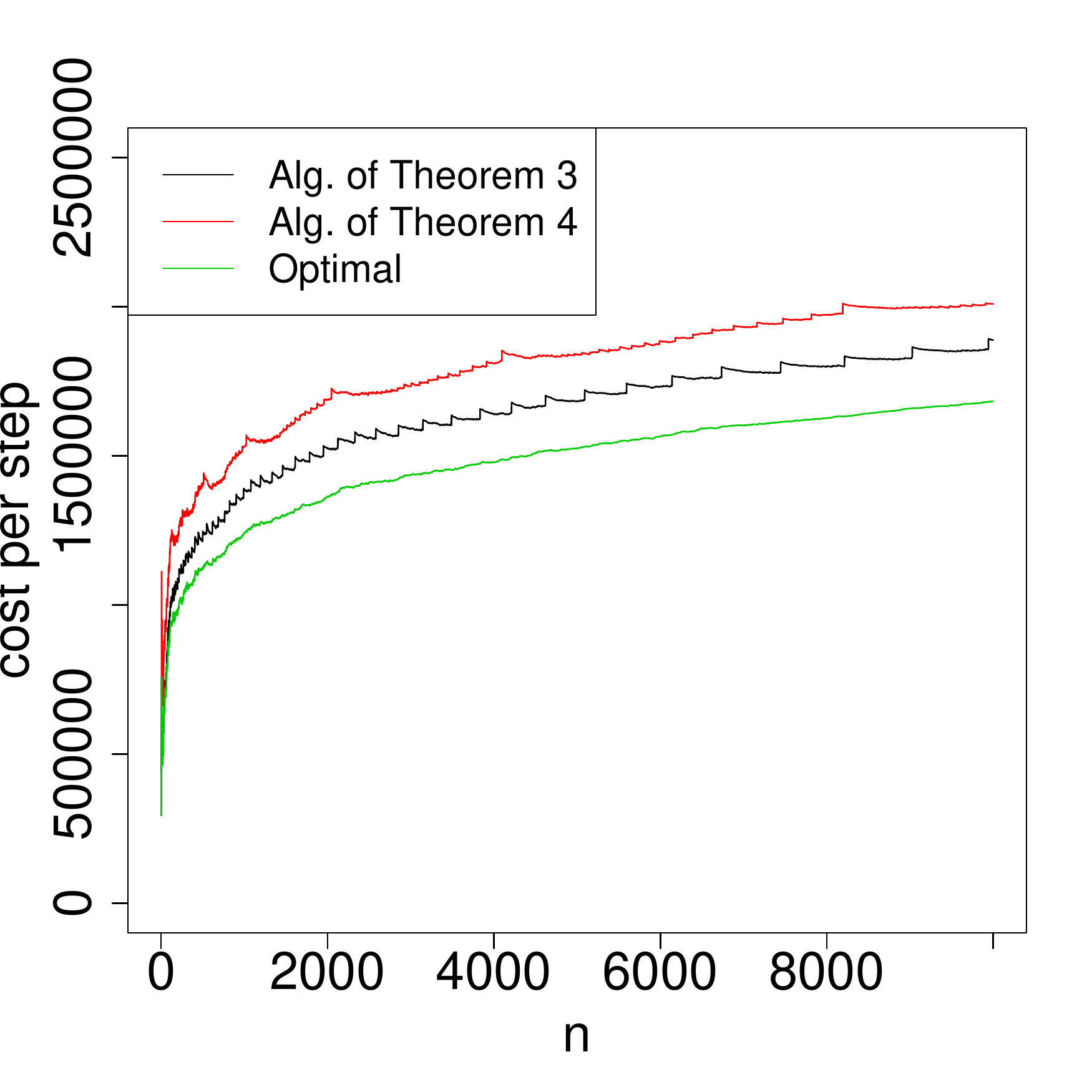}
    \caption{\LinearBmc with $\lbar/\rbar = .1$, $n\le 10,000$}\label{fig:p2 right}
  \end{subfigure}
  \begin{subfigure}{.5\textwidth}
    \centering
    \includegraphics[width=2.7in,height=\tmpG]{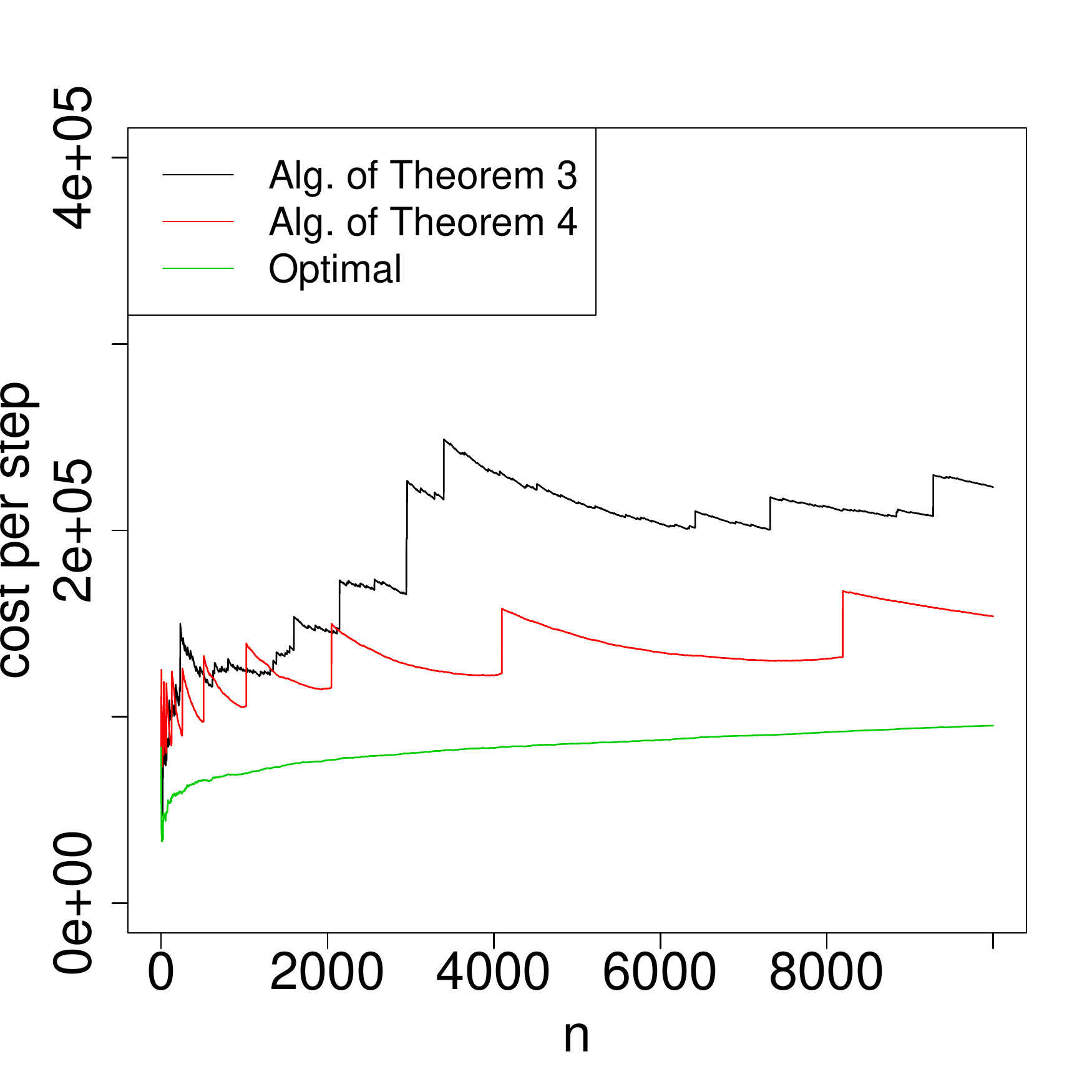}
    \caption{\LinearBmc with $\lbar/\rbar = 100$, $n\le 10,000$}\label{fig:p2 left}
  \end{subfigure}%
  \caption{Instances with $\mu=10$, $v=1$, and (a) $\lbar/\rbar$ small, and (b) $\lbar/\rbar$ large.}\label{fig:p2}
\end{figure}

\section{Acknowledgements}
Thanks to Mordecai Golin and Vagelis Hristidis for useful discussions.

{\small
\bibliographystyle{abbrv}
\bibliography{main}

\begin{thebibliography}{10}

\bibitem{alsubaiee2014asterixdb}
S.~Alsubaiee, Y.~Altowim, H.~Altwaijry, A.~Behm, V.~Borkar, Y.~Bu, M.~Carey,
  I.~Cetindil, M.~Cheelangi, K.~Faraaz, et~al.
\newblock {AsterixDB}: A scalable, open source {BDMS}.
\newblock {\em Proceedings of the VLDB Endowment}, 7(14):1905--1916, 2014.

\bibitem{arge_external-memory_2010}
L.~Arge and N.~Zeh.
\newblock External-memory algorithms and data structures.
\newblock In M.~J. Atallah and M.~Blanton, editors, {\em Algorithms and
  {Theory} of {Computation} {Handbook}}, pages 10--10. Chapman \& Hall/CRC,
  2010.

\bibitem{bentley_general_1982}
J.~L. Bentley and D.~J. Brown.
\newblock A general class of resource tradeoffs.
\newblock {\em Journal of Computer and System Sciences}, 25(2):214--238, Oct.
  1982.

\bibitem{cattell2011scalable}
R.~Cattell.
\newblock Scalable {SQL} and {NoSQL} data stores.
\newblock {\em ACM SIGMOD Record}, 39(4):12--27, 2011.

\bibitem{chang2008bigtable}
F.~Chang, J.~Dean, S.~Ghemawat, W.~C. Hsieh, D.~A. Wallach, M.~Burrows,
  T.~Chandra, A.~Fikes, and R.~E. Gruber.
\newblock Bigtable: A distributed storage system for structured data.
\newblock {\em ACM Trans. Comput. Syst.}, 26(2):4:1--4:26, June 2008.

\bibitem{choy_construction_1983}
D.~Choy and C.~Wong.
\newblock Construction of optimal $\alpha$---$\beta$ leaf trees with
  applications to prefix code and information retrieval.
\newblock {\em SIAM Journal on Computing}, 12(3):426--446, Aug. 1983.

\bibitem{corbett2013spanner}
J.~C. Corbett, J.~Dean, M.~Epstein, A.~Fikes, C.~Frost, J.~Furman, S.~Ghemawat,
  A.~Gubarev, C.~Heiser, P.~Hochschild, et~al.
\newblock Spanner: {Google's} globally distributed database.
\newblock {\em ACM Transactions on Computer Systems (TOCS)}, 31(3):8, 2013.

\bibitem{george_hbase:_2011}
L.~George.
\newblock {\em {HBase}: the definitive guide}.
\newblock O'Reilly Media, 2011.

\bibitem{ghosh_fast_2015}
M.~Ghosh, I.~Gupta, S.~Gupta, and N.~Kumar.
\newblock Fast compaction algorithms for {NoSQL} databases.
\newblock Technical report, University of Illinois, Dept. of Computer Science,
  Apr. 2015.

\bibitem{golin_more_2008}
M.~Golin and J.~Li.
\newblock More efficient algorithms and analyses for unequal letter cost
  prefix-free coding.
\newblock {\em IEEE Transactions on Information Theory}, 54(8):3412--3424, Aug.
  2008.

\bibitem{judd_scale_2008}
D.~Judd.
\newblock Scale out with {HyperTable}.
\newblock {\em Linux magazine, August 7th}, 2008.

\bibitem{kapoor_optimum_1989}
S.~Kapoor and E.~M. Reingold.
\newblock Optimum lopsided binary trees.
\newblock {\em J. ACM}, 36(3):573--590, July 1989.

\bibitem{kepner_achieving_2014}
J.~Kepner, W.~Arcand, D.~Bestor, B.~Bergeron, C.~Byun, V.~Gadepally,
  M.~Hubbell, P.~Michaleas, J.~Mullen, A.~Prout, A.~Reuther, A.~Rosa, and
  C.~Yee.
\newblock Achieving 100,000,000 database inserts per second using {Accumulo}
  and {D}4m.
\newblock In {\em 2014 {IEEE} {High} {Performance} {Extreme} {Computing}
  {Conference} ({HPEC})}, pages 1--6, Sept. 2014.

\bibitem{khetrapal_hbase_2006}
A.~Khetrapal and V.~Ganesh.
\newblock {HBase} and {Hypertable} for large scale distributed storage systems.
\newblock {\em Dept. of Computer Science, Purdue University}, pages 22--28,
  2006.

\bibitem{patil_ycsb++:_2011}
S.~Patil, M.~Polte, K.~Ren, W.~Tantisiriroj, L.~Xiao, J.~López, G.~Gibson,
  A.~Fuchs, and B.~Rinaldi.
\newblock {YCSB}++: benchmarking and performance debugging advanced features in
  scalable table stores.
\newblock In {\em Proceedings of the 2nd {ACM} {Symposium} on {Cloud}
  {Computing}}, page~9. ACM, 2011.

\bibitem{redmond2012seven}
E.~Redmond and J.~R. Wilson.
\newblock {\em Seven databases in seven weeks: a guide to modern databases and
  the {NoSQL} movement}.
\newblock Pragmatic Bookshelf, 2012.

\bibitem{sniedovich2003or}
M.~Sniedovich.
\newblock {OR/MS} {Games}: 4. {The} joy of egg-dropping in {Braunschweig} and
  {Hong} {Kong}.
\newblock {\em INFORMS Transactions on Education}, 4(1):48--64, 2003.

\bibitem{strauch2011nosql}
C.~Strauch.
\newblock {NoSQL} databases.
\newblock {\em Lecture Notes, Stuttgart Media University}, 2011.

\bibitem{vitter_external_2001}
J.~S. Vitter.
\newblock External memory algorithms and data structures: dealing with massive
  data.
\newblock {\em ACM Comput. Surv.}, 33(2):209--271, June 2001.

\end{thebibliography}
}

\end{document}